%% file: main.tex
\documentclass{article}

\pdfoutput=1

\usepackage{arxiv}

\usepackage[utf8]{inputenc} 
\usepackage[T1]{fontenc}    
\usepackage{hyperref}       
\usepackage{url}            
\usepackage{booktabs}       
\usepackage{amsfonts}       
\usepackage{nicefrac}       
\usepackage{microtype}      
\usepackage{lipsum}		
\usepackage{graphicx}
\usepackage{natbib}
\usepackage{doi}

\usepackage{amsmath}
\usepackage{amssymb}
\usepackage{amsthm}

\newtheorem{definition}{Definition}
\newtheorem{theorem}{Theorem}
\newtheorem{corollary}{Corollary}
\newtheorem{proposition}{Proposition}
\newtheorem{remark}{Remark}
\newtheorem{lemma}{Lemma}

\newcommand{\triple}[3]{
  \expandafter\newcommand\csname #3\endcsname[3]{%
    \ensuremath{#1##1#2\,##2\,#1##3#2}%
  }
} 
\triple{\{}{\}}{hoare} 
\triple{[}{]}{incre} 
\triple{\langle}{\rangle}{lisbon} 
\triple{(}{)}{necpre} 

\newcommand{\psqhoare}{\hoare{P}{S}{Q} }
\newcommand{\psqincr}{\incre{P}{S}{Q} }
\newcommand{\psqlisbon}{\lisbon{P}{S}{Q} }
\newcommand{\psqnecpre}{\necpre{P}{S}{Q} }

\usepackage{silence}
\title{Combining Axiomatic Program Logics with Refinement Proofs}

\author{ \href{https://orcid.org/0000-0002-0734-7969}{\includegraphics[scale=0.06]{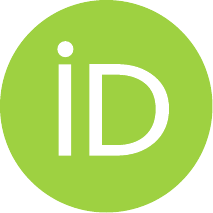}\hspace{1mm}S\"uha Orhun Mutluergil} \\
	Computer Science and Engineering Program\\
	Sabancı University\\
	İstanbul, Türkiye \\
	\texttt{suha.mutluergil@sabanciuniv.edu} \\
	\And
	\href{https://orcid.org/0009-0005-5316-7845}{\includegraphics[scale=0.06]{orcid.pdf}\hspace{1mm}Alperen Do\u{g}an} \\
	Computer Science and Engineering Program\\
	Sabancı University\\
	İstanbul, Türkiye \\
	\texttt{alperend@sabanciuniv.edu} \\
}

\hypersetup{
pdftitle={Combining Axiomatic Program Logics with Refinement Proofs},
pdfsubject={cs.LO},
pdfauthor={S\"uha Orhun Mutluergil, Alperen Do\u{g}an},
pdfkeywords={functional verification, relational verification, simulation relations},
}

\begin{document}
\maketitle

\begin{abstract}
    Refinement proofs verify an implementation by showing that its behaviours are subsumed by a simpler specification, on which functional properties are easier to establish. We study how such proofs interact with the axiomatic program logics used to verify the specification. We first give a uniform account of Hoare, Incorrectness, Lisbon, and Necessary Preconditions logic, classified by the direction in which each constrains a transition and by whether it over- or under-approximates its target set. We then show that simulation relations transfer state-based safety properties: a forward simulation carries a Hoare (inductive) invariant of the specification to one of the implementation, and forward and backward simulations both carry ordinary invariants, via the pre-image of the relation. Finally, we characterize, within these logics, when a relation is a simulation, forward simulations by the validity of Hoare or Lisbon triples, backward simulations by Necessary Preconditions or Incorrectness triples, so that the simulation obligation reduces to a triple in a functional logic. We illustrate the development with a concurrent counter, carrying a safety bound from an atomic sequential specification to a Left--Right implementation through an intermediate nondeterministic concurrent counter, with a forward simulation on one side and a backward simulation on the other.
\end{abstract}


\section{Introduction}

\input{arxivintroduction}

\section{Background}
\label{sec:background}

\input{arxivbackground}

\section{Preliminaries}
\label{sec:preliminaries}

\input{arxivpreliminaries}

\section{Axiomatic Models of Forward and Backward Simulations}
\label{sec:axiomatic-sim}

\input{arxivaxiomaticsim}

\section{Checking Validity of Forward and Backward Simulations}
\label{sec:validity-sim}

\input{arxivvaliditysim}

\section{Conclusions}

\input{arxivconclusions}

\clearpage

\bibliographystyle{unsrtnat}
\bibliography{references}

\clearpage

\begin{appendix}
    \input{arxivappendix.tex}
\end{appendix}

\end{document}

%% file: arxivintroduction.tex
Ensuring the correctness of safety-critical software increasingly relies on formal verification. Functional correctness is commonly expressed in axiomatic program logics: Hoare logic~\citep{Hoare69}, which over-approximates the reachable states of a program; Incorrectness logic~\citep{OHearn2019:incorrectness}, which under-approximates them. The more recent Lisbon~\citep{Ascari23:sil} and Necessary Preconditions~\citep{CousotNecPre} logic reason in the backward direction. Verifying a property directly on an implementation, however, is costly, and the cost is paid again after every change to the code. Refinement avoids this: instead of reasoning about the implementation, one proves the property on a simpler abstract specification and shows that every behavior of the implementation is permitted by the specification~\cite{LynchV95:forward}. Because specifications are stated over clean mathematical structures, the functional proof is usually easier.
 
The standard techniques for relating an implementation to a specification is a \emph{simulation relation}~\cite{LynchV95:forward}, and the standard technique for transferring a property is to pull it back along that relation. This paper studies the interaction between the two. We ask two questions: under forward or backward simulations, what a property proved on the specification guarantees on the implementation; and  how the simulation itself can be discharged using the same functional logics that verify the specification.
 
To keep the development concrete, we carry a single running example throughout: a shared counter supporting concurrent increment and read operations. The property of interest, that no read ever reports more increments than have been invoked, is immediate on an atomic sequential counter, and we transport it down to a Left--Right (LR) concurrent implementation along simulation relations. No single simulation links the implementation to the sequential specification, so the proof factors through an intermediate nondeterministic concurrent (NDC) counter, with a forward simulation on one side and a backward simulation on the other.
 
Section~\ref{sec:background} introduces the three machines and the two relations informally. Section~\ref{sec:preliminaries} describes the simulation framework, the functional logics, and the notion of invariant we transfer. Section~\ref{sec:axiomatic-sim} proves the transfer theorems and applies them to the counter. Section~\ref{sec:validity-sim} gives the simulation validity characterizations, and the appendices discharge the example's forward and backward simulations using Lisbon and Incorrectness triples.

%% file: arxivbackground.tex
We illustrate the development with a single running example: a shared counter supporting concurrent \emph{increment} and \emph{read} operations. Our goal is to verify a safety property of a particular implementation, that no read ever reports more increments than that have actually been invoked. Implementations of concurrent objects are, however, awkward to reason about directly. The strategy of this paper is therefore to prove the property on a simple \emph{specification}, where it is almost immediate, and to transport it down to the implementation along simulation relations. As we will see, no single simulation links the implementation to the simplest specification; the proof factors through an intermediate machine. This section introduces the three machines and the two relations informally and the remaining sections make the reasoning axiomatic.

\subsection{A Concurrent Counter: the Left--Right Implementation}
\label{sec:lr}
 
The implementation we wish to verify, the \emph{Left--Right} (LR) counter, keeps the count split across two registers $L$ and $R$. An increment bumps one of the two registers, chosen nondeterministically, modeling a counter split across two threads or nodes. A read is performed in two steps: at \emph{begin} it snapshots $L$ into a local variable $v_L$, and at \emph{return} it adds the \emph{current} value of $R$ and reports $v_L + R$. This split makes the implementation efficient but subtle: the value a read returns is not determined when it begins, but only when it returns, and it depends on every increment that touches $R$ in between.
 
Throughout, $\mathbb{I}_I$ and $\mathbb{I}_R$ are the (disjoint) sets of increment and read operation identifiers, and $\mathbb{I} = \mathbb{I}_I \cup \mathbb{I}_R$. Each operation identifier carries a program counter that moves from \textsc{init} to \textsc{done} (through \textsc{mid} for the two step operations); the maps $i_{\bullet}$ and $o_{\bullet}$ record the input and output values of each operation and constitute the externally visible interface; the register updates are internal.

\begin{definition}[Left--Right counter]
\label{def:lr}
The state of the LR counter consists of $L \in \mathbb{N}$, $R \in \mathbb{N}$, a program counter $\mathit{pc}_{\mathit{lr}} : \mathcal{I} \to \{\textsc{init},\textsc{mid},\textsc{done}\}$, a snapshot map $v_L : \mathcal{I}_R \to \mathbb{N}$, a set of old increment values $\mathit{old}_{lr} \subseteq \mathbb{N}$, and interface maps $i_{\mathit{lr}}, o_{\mathit{lr}}$. Its operations are
\[
\begin{array}{l}
\textsf{increment}(id):\\
\quad \textbf{assume } \mathit{pc}_{\mathit{lr}}(id) = \textsc{init}\\
\quad \textbf{if } (*)\ \ L \leftarrow L + 1\ \ \textbf{else}\ \ R \leftarrow R + 1\\
\quad \mathit{old}_{lr} \leftarrow \mathit{old}_{lr} \cup \{L+R\} \\
\quad \mathit{pc}_{\mathit{lr}}(id) \leftarrow \textsc{done}\\[4pt]
\textsf{begin\_read}(id):\\
\quad \textbf{assume } \mathit{pc}_{\mathit{lr}}(id) = \textsc{init}\\
\quad v_L(id) \leftarrow L\\
\quad \mathit{pc}_{\mathit{lr}}(id) \leftarrow \textsc{mid}\\[4pt]
\textsf{return\_read}(id):\\
\quad \textbf{assume } \mathit{pc}_{\mathit{lr}}(id) = \textsc{mid}\\
\quad \mathit{pc}_{\mathit{lr}}(id) \leftarrow \textsc{done}\\
\quad o_{\mathit{lr}}(id) \leftarrow v_L(id) + R\\
\quad \textbf{return } v_L(id) + R
\end{array}
\]
\end{definition}

The correctness property of interest is a safety bound: for every read $id$ that has completed, the reported value $o_{\mathit{lr}}(id)$ is at most $L + R$, the number of increments performed so far. Stating and proving this directly on the LR counter means reasoning about all interleavings of the split reads with concurrent increments to $R$, which is precisely the kind of global argument we prefer to avoid.

\subsection{Transferring Correctness from a Sequential Specification}
\label{sec:seq}
 
An idealized \emph{sequential} counter, in which each operation executes atomically, makes the property trivial: a single variable $x$ holds the count, an increment adds one to it atomically, and a read atomically returns $x$. The safety bound is then trivial, because $x$ \emph{is} the number of increments.
 
\begin{definition}[Sequential counter]
\label{def:seq}
The state consists of $x \in \mathbb{N}$, a program counter $\mathit{pc}_s : \mathcal{I} \to \{\textsc{init},\textsc{done}\}$, a set of old increment values $\mathit{old}_{s} \subseteq \mathbb{N}$, and interface maps $i_s, o_s$. Its operations are
\[
\begin{array}{l}
\textsf{increment}(id):\\
\quad \textbf{assume } \mathit{pc}_s(id) = \textsc{init}\\
\quad x \leftarrow x + 1\\
\quad \mathit{old}_s \leftarrow \mathit{old}_s \cup \{x\} \\
\quad \mathit{pc}_s(id) \leftarrow \textsc{done}\\[4pt]
\textsf{read}(id):\\
\quad \textbf{assume } \mathit{pc}_s(id) = \textsc{init}\\
\quad o_s(id) \leftarrow x\\
\quad \mathit{pc}_s(id) \leftarrow \textsc{done}\\
\quad \textbf{return } x
\end{array}
\]
\end{definition}
 
Were the LR counter to refine the sequential counter, the safety bound would transfer downward: a simulation pulls an invariant of the specification back to an invariant of the implementation. The difficulty is finding the witnessing simulation. A forward simulation from LR to the sequential counter would have to match every implementation step against specification steps using only the past and present. But the sequential read is atomic: it must commit to a returned value at one instant, whereas the LR read fixes its value only at \textsf{return\_read}, after an unknown number of interleaving increments to $R$. At \textsf{begin\_read} there is no way to know which sequential state to move to, because the answer depends on the future. A single forward simulation between the two therefore does not exist, and this is the look ahead that, in general, forces backward (prophecy style) reasoning.
 
Rather than search for a more delicate single relation, we factor the proof through an intermediate specification. This is guaranteed by the completeness of forward and backward simulations: when one machine implements another, the trace inclusion can always be witnessed by a forward simulation into some intermediate automaton followed by a backward simulation out of it \citep{LynchV95:forward}. The next machine is an instance of that intermediate automaton.

\subsection{An Intermediate Specification: the Nondeterministic Concurrent Counter}
\label{sec:ndc}
 
The \emph{nondeterministic concurrent} (NDC) counter is designed to be loose. It does not track a single count; instead it records, for each read, which increments that read is obliged to observe and which it is merely permitted to observe. It maintains the set $P$ of \emph{pending} increments (begun but not finished), the set $C$ of \emph{completed} increments, and the set $A$ of \emph{active} reads. When a read begins, it snapshots the completed increments into its \emph{must-see} set $\mathit{MS}$ and the pending increments into its \emph{may-see} set $\mathit{YS}$. When the read returns, it nondeterministically chooses any set $S$ between these two bounds and reports $|S|$. This freedom lets the NDC counter match the LR counter step for step: a value that the split LR read produces at return time is always one of the admissible $|S|$. At the same time the must-see/may-see bookkeeping allows NDC counter to be related to the atomic sequential counter.
 
\begin{definition}[Nondeterministic concurrent counter]
\label{def:ndc}
The state consists of $P \subseteq \mathcal{I}_I$ (pending increments), $C \subseteq \mathcal{I}_I$ (completed increments), $A \subseteq \mathcal{I}_R$ (active reads), maps $\mathit{MS}, \mathit{YS} : \mathcal{I}_R \to \mathcal{P}(\mathcal{I}_I)$ (must-see and may-see), a program counter $\mathit{pc}_n : \mathcal{I} \to \{\textsc{init},\textsc{mid},\textsc{done}\}$, a set of old increment values $\mathit{old}_{n} \subseteq \mathbb{N}$, and interface maps $i_n, o_n$. Its operations are
\[
\begin{array}{l}
\textsf{begin\_increment}(id):\\
\quad \textbf{assume } \mathit{pc}_n(id) = \textsc{init},\ id \notin P,\ id \notin C\\
\quad P \leftarrow P \cup \{id\}\\
\quad \mathit{old}_n \leftarrow \mathit{old}_n \cup \{|C|+|P|\} \\
\quad \textbf{for all } t \in A:\ \mathit{YS}(t) \leftarrow \mathit{YS}(t) \cup \{id\}\\
\quad \mathit{pc}_n(id) \leftarrow \textsc{mid}\\[4pt]
\textsf{return\_increment}(id):\\
\quad \textbf{assume } \mathit{pc}_n(id) = \textsc{mid},\ id \in P,\ id \notin C\\
\quad P \leftarrow P \setminus \{id\};\quad C \leftarrow C \cup \{id\}\\
\quad \mathit{pc}_n(id) \leftarrow \textsc{done}\\[4pt]
\textsf{begin\_read}(id):\\
\quad \textbf{assume } \mathit{pc}_n(id) = \textsc{init},\ id \notin A\\
\quad A \leftarrow A \cup \{id\}\\
\quad \mathit{MS}(id) \leftarrow C;\quad \mathit{YS}(id) \leftarrow P\\
\quad \mathit{pc}_n(id) \leftarrow \textsc{mid}\\[4pt]
\textsf{return\_read}(id):\\
\quad \textbf{assume } \mathit{pc}_n(id) = \textsc{mid},\ id \in A\\
\quad A \leftarrow A \setminus \{id\}\\
\quad \textbf{choose } S \text{ with } \mathit{MS}(id) \subseteq S \subseteq \mathit{MS}(id) \cup \mathit{YS}(id)\\
\quad o_n(id) \leftarrow |S|\\
\quad \mathit{pc}_n(id) \leftarrow \textsc{done}\\
\quad \textbf{return } |S|
\end{array}
\]
\end{definition}

\subsection{The Simulation Chain}
\label{sec:chain}
 
We relate the three machines by
\[
  \text{LR}
  \;\xrightarrow{\;\text{forward}\;}\;
  \text{NDC}
  \;\xrightarrow{\;\text{backward}\;}\;
  \text{Seq},
\]
which lets the composed simulation relation carry the sequential counter's Hoare invariant, the bound stating that every completed read reports at most $L + R$, back to the LR counter as a state predicate. The two relations are given below; both require the interface maps to agree, forcing the same externally visible behaviour.

\subsubsection{Forward Simulation from LR to NDC}
\label{sec:fw-rel}
 
Let the concrete (LR) state be $\sigma = (L, R, \mathit{pc}_{\mathit{lr}}, v_L, i_{\mathit{lr}}, o_{\mathit{lr}})$ and the abstract (NDC) state be $\tau = (P, C, A, \mathit{MS}, \mathit{YS}, \mathit{pc}_n, i_n, o_n)$. The forward simulation relation $R(\sigma, \tau)$ is the conjunction of
\begin{align}
&|C| = L + R \tag{F1}\\
&\forall id \in \mathcal{I}.\ \mathit{pc}_{\mathit{lr}}(id) = \textsc{mid} \iff id \in A \tag{F2}\\
&\forall id \in A.\ |\mathit{MS}(id)| \le v_L(id) + R \le |\mathit{MS}(id) \cup \mathit{YS}(id)| \tag{F3}\\
&\mathit{old}_n \subseteq \mathit{old}_{\mathit{lr}} \tag{F4}\\
&\forall id \in \mathcal{I}.\ \mathit{pc}_n(id) \neq \textsc{init} \implies \mathit{pc}_{\mathit{lr}}(id) \neq \textsc{init} \wedge i_{\mathit{lr}}(id) = i_n(id) \tag{F5}\\
&\forall id \in \mathcal{I}.\ \mathit{pc}_{\mathit{lr}}(id) = \textsc{done} \implies \mathit{pc}_n(id) = \textsc{done} \wedge o_{\mathit{lr}}(id) = o_n(id) \tag{F6}
\end{align}
 
The intuition is an \emph{action correspondence} that matches each implementation step with a sequence of specification steps.
\begin{itemize}
  \item An LR \textsf{increment}, a single atomic step that bumps $L$ or $R$, is matched by the pair \textsf{begin\_increment}; \textsf{return\_increment}, which moves one identifier into $C$. Either way the count rises by one, so $|C| = L+R$ is restored~(F1).
  \item An LR \textsf{begin\_read} is matched by an NDC \textsf{begin\_read}: the snapshot $v_L \leftarrow L$ on the concrete side and $\mathit{MS} \leftarrow C$, $\mathit{YS} \leftarrow P$ on the abstract side establish the two bounds of~(F3) for the new read, with $id$ entering $A$ exactly as $\mathit{pc}_{\mathit{lr}}(id)$ enters \textsc{mid}~(F2).
  \item An LR \textsf{return\_read} reports $v_L(id) + R$. The matching NDC \textsf{return\_read} chooses $S$ with $|S| = v_L(id) + R$; invariant~(F3) guarantees such an $S$ lies between $\mathit{MS}(id)$ and $\mathit{MS}(id) \cup \mathit{YS}(id)$, so the two machines report the same value and (F6) is preserved.
\end{itemize}
Conditions (F5) and (F6) keep the inputs and outputs of corresponding operations equal, so equal observable behaviour is maintained throughout. Because every LR step is answered forward in this way, $R$ is a forward simulation and $\mathit{traces}(\text{LR}) \subseteq \mathit{traces}(\text{NDC})$.
 
\subsubsection{Backward Simulation from NDC to the Sequential Counter}
\label{sec:bw-rel}
 
Now let the concrete (NDC) state be $\sigma = (P, C, A, \mathit{MS}, \mathit{YS}, \mathit{pc}_n, i_n, o_n)$ and the abstract (sequential) state be $\tau = (x, \mathit{pc}_s, i_s, o_s)$. The backward simulation relation $R(\sigma, \tau)$ is the conjunction of
\begin{align}
&|C| \le x \le |C| + |P| \tag{B1}\\
&\mathit{old}_s \subseteq \mathit{old}_n \tag{B2}\\
&\forall id \in \mathcal{I}.\ \mathit{pc}_s(id) \neq \textsc{init} \implies \mathit{pc}_n(id) \neq \textsc{init} \wedge i_n(id) = i_s(id) \tag{B3}\\
&\forall id \in \mathcal{I}.\ \mathit{pc}_n(id) = \textsc{done} \implies \mathit{pc}_s(id) = \textsc{done} \wedge o_n(id) = o_s(id) \tag{B4}
\end{align}
 
The prophecy is expressed by invariant~(B1): the sequential count $x$ is pinned only to the interval between the completed increments and the completed plus pending increments, leaving room to decide \emph{when} an abstract increment takes effect. Because the simulation is backward, that decision is made by looking at the post-state. The action correspondence runs as follows.
\begin{itemize}
  \item Across an NDC \textsf{begin\_increment}; \textsf{return\_increment} pair, exactly one atomic sequential \textsf{increment} must fire. Which of the two concrete steps it accompanies, or whether the abstract machine stutters, is the prophecy, resolved by looking at the post-state so that the sequential count x stays inside the interval (B1); this is the look ahead that a forward simulation could not perform.
  \item An NDC read reports $|S|$, a value between its must-see count $|MS(id)|$ and its may-see count $|MS(id) \cup YS(id)|$. The matching atomic sequential read takes effect at the read's linearization point: the step at which the running count x first equals $|S|$. Because |S| may be smaller than the number of increments completed by the time the read returns, this point generally lies strictly before $\textsf{return\_read}$, at $\textsf{begin\_read}$, or at one of the increments interleaved with the read. Invariant (B1), together with the must-/may-see bounds, guarantees the count passes through $|S|$ while the read is active, so such a step exists; there both machines report $|S|$ and (B4) is established.
  \item Every NDC step at which neither the increment nor a read takes effect, in particular \textsf{return\_read}, whose value was already fixed at the linearization point, and \textsf{begin\_read} for a read that linearizes later, is matched by a sequential stutter.
\end{itemize}
Since every NDC step admits such a matching (or stuttering) abstract step constructed from the future, $R$ is a backward simulation and $\mathit{traces}(\text{NDC}) \subseteq \mathit{traces}(\text{Seq})$.
 
\medskip
 
Composing the two simulations relates the LR counter to the sequential counter, so the safety bound proved on the sequential counter holds of the LR implementation. The forward half pins down how the implementation's choices are reflected in the intermediate machine, and the backward half supplies the look ahead that relates the intermediate machine to the atomic specification. The rest of the paper recasts both halves axiomatically and shows how their validity can be discharged using the functional program logics.

%% file: arxivpreliminaries.tex
\subsection{Refinement Proofs}
\label{sec:refinement}

\subsubsection{Automata, Executions, and Traces}
\label{sec:automata}

We use the labelled transition system model of ~\cite{LynchV95:forward}. An \emph{automaton} $A$ consists of a set $\mathit{states}(A)$ of states, a nonempty set $\mathit{start}(A) \subseteq \mathit{states}(A)$ of start states, a set $\mathit{acts}(A)$ of actions containing a distinguished \emph{internal} action $\tau$, and a set $\mathit{steps}(A) \subseteq \mathit{states}(A) \times \mathit{acts}(A) \times \mathit{states}(A)$ of transitions. The \emph{external} actions are $\mathit{ext}(A) = \mathit{acts}(A) \setminus \{\tau\}$; these are the actions visible to an observer. We write $s \xrightarrow{a}_A s'$ for $(s,a,s') \in \mathit{steps}(A)$, and $\widehat{\beta}$ for the sequence obtained from a sequence of actions $\beta$ by deleting all $\tau$ events.

An \emph{execution fragment} of $A$ is a finite or infinite alternating
sequence $\alpha = s_0\, a_1\, s_1\, a_2\, s_2 \cdots$ of states and
actions with $s_i \xrightarrow{a_{i+1}}_A s_{i+1}$ for all $i$; it is an
\emph{execution} if $s_0 \in \mathit{start}(A)$. A state is
\emph{reachable} if it is the last state of some finite execution, and
we write $\mathit{reach}(A)$ for the set of reachable states. The
\emph{trace} of $\alpha$ is $\mathit{trace}(\alpha) =
\widehat{a_1 a_2 \cdots}$, the sequence of external actions along
$\alpha$, and $\mathit{traces}(A)$ is the set of all traces. For a
finite sequence $\beta$ over $\mathit{ext}(A)$ we write
$s \xrightarrow{\beta}_A s'$ when $A$ has a finite execution fragment
from $s$ to $s'$ with trace $\beta$; this ``weak'' arrow absorbs
internal steps and lets the specification answer one implementation step
with a sequence of its own steps.

The behavioral comparison of interest is \emph{trace inclusion}. For automata $S_1$ and $S_2$ with $\mathit{ext}(S_1) = \mathit{ext}(S_2)$ we write
\[
  S_1 \sqsubseteq_T S_2 \quad\stackrel{\triangle}{\Longleftrightarrow}\quad
  \mathit{traces}(S_1) \subseteq \mathit{traces}(S_2),
\]
read ``$S_1$ implements $S_2$''. Throughout, $S_1$ is the implementation (concrete) automaton and $S_2$ the specification (abstract) one.

\subsubsection{Forward Simulations}
\label{sec:forward}
 
A \emph{forward simulation} from $S_1$ to $S_2$ is a relation $f \subseteq \mathit{states}(S_1) \times \mathit{states}(S_2)$ such that
\begin{enumerate}
  \item if $s \in \mathit{start}(S_1)$ then
        $f[s] \cap \mathit{start}(S_2) \neq \emptyset$, and
  \item if $s \xrightarrow{a}_{S_1} s'$ and $u \in f[s]$, then there exists
        $u' \in f[s']$ such that $u \xrightarrow{\widehat{a}}_{S_2} u'$,
\end{enumerate} 
where $f[s] = \{\, u \mid (s,u) \in f \,\}$. We write $S_1 \leq_F S_2$ when such an $f$ exists. The defining picture is forward in time: starting from a related pair $(s,u)$, every step of the implementation can be answered by a matching move of the specification leading to a still related pair \cite{LynchV95:forward}.
 
\begin{theorem}[Soundness, {\citep{LynchV95:forward}}]
\label{thm:fw-sound}
    $S_1 \leq_F S_2 \implies S_1 \sqsubseteq_T S_2$.
\end{theorem}
 
\begin{theorem}[Partial completeness, {\citep{LynchV95:forward}}]
\label{thm:fw-complete}
    If $S_2$ is deterministic and $S_1 \sqsubseteq_T S_2$, then $S_1 \leq_F S_2$.
\end{theorem}
 
Theorem~\ref{thm:fw-complete} also exposes the limitation of forward simulations: completeness is only \emph{partial}, and the determinism of $B$ is essential. When the specification is nondeterministic, trace inclusion need not be witnessed by any forward simulation. The obstruction is temporal. A forward simulation must, at each step, commit to a matching specification state on the basis of the past and present alone, whereas a correct implementation may resolve a nondeterministic choice \emph{later} than the specification does. This is the typical situation in linearizability arguments, where the abstract effect of an operation is fixed only once some future event occurs.

\subsubsection{Backward Simulations}
\label{sec:backward}

Backward simulations are the dual remedy, propagating the relation backward in time so that ``future'' choices can be matched. A \emph{backward simulation} from $S_1$ to $S_2$ is a \emph{total} relation $b \subseteq \mathit{states}(S_1) \times \mathit{states}(S_2)$ such that
\begin{enumerate}
  \item if $s \in \mathit{start}(S_1)$ then $b[s] \subseteq \mathit{start}(S_2)$, and
  \item if $s \xrightarrow{a}_{S_1} s'$ and $u' \in b[s']$, then there
        exists $u \in b[s]$ such that $u \xrightarrow{\widehat{a}}_{S_2} u'$.
\end{enumerate}
We write $S_1 \leq_B S_2$ when such a $b$ exists, and $S_1 \leq_{iB} S_2$ when moreover $b$ is \emph{image-finite}, i.e.\ $b[s]$ is finite for every $s$. The duality with forward simulations is not perfect. Two things differ: the totality requirement and the image-finite condition. From any state the histories are finite but the futures may be infinite, so reasoning backward needs an extra finiteness assumption to be sound for infinite behaviors \cite{LynchV95:forward}.
 
\begin{theorem}[Soundness, {\cite{LynchV95:forward}}]
\label{thm:bw-sound}
    $S_1 \leq_B S_2$ implies inclusion of \emph{finite} traces, and $S_1 \leq_{iB} S_2 \implies S_1 \sqsubseteq_T S_2$ (all traces, including infinite ones).
\end{theorem}
 
For finite traces a backward simulation is sound with no further hypotheses; for infinite traces image-finiteness lets one assemble the infinitely many local matches into a single infinite specification execution, by a K\"onig's-lemma argument \cite{LynchV95:forward}. Backward simulations are likewise only partially complete, with the structural restriction now on the implementation rather than the specification.
 
\begin{theorem}[Partial completeness, {\cite{LynchV95:forward}}]
\label{thm:bw-complete}
    If $S_1$ is a forest and $S_1 \sqsubseteq_T S_2$, then $S_1 \leq_B S_2$; moreover, if $S_2$ has finite invisible nondeterminism then $S_1 \leq_{iB} S_2$.
\end{theorem}

Neither forward nor backward simulations are complete on their own: each is guaranteed to exist only under the structural hypotheses of Theorems~\ref{thm:fw-complete} and~\ref{thm:bw-complete}. Their \emph{combination}, however, is complete. \cite{LynchV95:forward} show that whenever $S_1 \sqsubseteq_T S_2$ there is an intermediate automaton $S_3$ with a forward simulation from $S_1$ to $S_3$ and a backward simulation from $S_3$ to $S_2$ (image-finite when $S_2$ has finite invisible nondeterminism). Thus a forward simulation, followed by a backward simulation, can witness any provable trace inclusion.

\subsubsection{History and Prophecy Variables}
\label{sec:history-prophecy}
 
This completeness has a familiar reading in terms of auxiliary state. A \emph{history relation} is a forward simulation whose inverse is a refinement, and a \emph{prophecy relation} is a backward simulation whose inverse is a refinement \citep{LynchV95:forward}; these are the abstract, relational counterparts of the auxiliary (history) variables of \cite{OwickiGries1976} and the prophecy variables of \cite{AbadiLamport91:refinement}. A history variable records information about the \emph{past} of an execution without affecting its behaviour, and adding one is what is needed when the only obstacle to a refinement mapping is missing memory; this is the forward direction. A prophecy variable nondeterministically \emph{guesses} information about the future and later checks the guess, letting the specification anticipate a choice the implementation has not yet made; this is the backward direction. The classical completeness theorem of \cite{AbadiLamport91:refinement} states that, under machine closure, finite invisible nondeterminism, and internal continuity, trace inclusion always implies the existence of a refinement mapping once the implementation is augmented with suitable history and prophecy variables; the same statement, through the correspondence above, as the completeness of forward-then-backward simulation. This is the theoretical justification for using a forward simulation, a backward simulation, or a composition of the two as the simulation step of a refinement proof, and it is why the development in the remainder of the paper treats both styles.

\begin{remark}[From trace inclusion to invariants]
\label{rem:traces-to-invariants}
    Although the soundness and completeness results above are phrased in terms of trace inclusion, the safety properties we transfer in this paper are \emph{sets of states} rather than sets of traces (Section~\ref{sec:invariants}). The two views meet through reachability: a forward simulation $f$ from $S_1$ to $S_2$ relates every reachable state of $S_1$ to some reachable state of $S_2$, so an invariant of $S_2$ pulls back, along $f$, to an invariant of $S_1$. This state level transfer, and its analogue for backward simulations, is the form actually used in Sections~\ref{sec:axiomatic-sim} onward, where the simulations are recast in the logic of actions and the transfer is carried out with the program logics of Section~\ref{sec:functionalLogics}.
\end{remark}

\subsection{Axiomatic Models}
\label{sec:functionalLogics}

There are new variants of Hoare logic. First, we would like to classify them. Program logics we consider define axiomatic semantics on imperative programs. Hence, we assume that actions consist of triples of the form: a unary predicate on pre-states, a binary predicate representing the program transformation and a unary predicate over post-states.

If a program logic fixes pre-states ( resp., post-states) and puts some restrictions on post-states (resp., pre-states) then it is called a \emph{forward} (resp., \emph{backward}) logic. If the restriction requires being a subset (resp., superset) of some target set, this logic is called an under-approximation (resp., over-approximation) logic. Target sets are either exact reachable post-states of a program from a given pre-state set or exact and all pre-states needed by the program to reach a given post-state set. To define these notions formally, we have to define the image and pre-image of a binary relation.

\begin{definition}[Image and Pre-Image]
Let $R \subseteq \Sigma \times \Theta$ be a binary relation between two state spaces $\Sigma$ and $\Theta$.

\emph{Image} of the set $P \subseteq \Sigma$ is defined as
\[
    R(P) \;\triangleq\; \{ \sigma' \in \Theta \mid \exists \sigma. P(\sigma) \wedge R(\sigma, \sigma') \}
\]
and  \emph{pre-image} of the set $Q \subseteq \Theta$ is defined as
\[
R^{-1}(Q) \;\triangleq\;
\{\;\sigma \in \Sigma \mid \exists \sigma'.\; Q(\sigma') \wedge R(\sigma,\sigma')\;\}.
\]

\end{definition}

If $S \subseteq \Sigma \times \Sigma $ is the binary predicate\footnote{Throughout the paper, we abuse the notation by using set and predicate notions in an interchanging way. For a unary predicate $P(\sigma)$, we assume the set $P = \{ \sigma \mid P(\sigma) \}$ and vice versa. Similarly, binary predicates correspond to binary relations over sets. Predicates take their arguments from base program state sets like $\Sigma$ and $\Theta$.} representing a program transformation, then $S(P)$ and $S^{-1}(P)$ represent the strongest post-condition and the strongest precondition of $S$ for a set of states $P \subseteq \Sigma$, respectively. \footnote{Terms strongest postcondition and precondition are introduced with respect to semantics of Hoare Logic. They might be confusing when used for other logics. Hence, we will stick to images and preimages from this point on.}.

Using these definitions, one can define four program logics\footnote{Actually, there can be more. We deliberately omit the discussion on termination and non-determinism. Depending on what will happen to non-terminating states and multiple states reachable from a single state  in the presence of non-determinism or multiple states terminating at the same target state, one can define different program logics. Decisions that logics we consider take on these matters will be clear when we formally define them later. See \citep{VerschtK25:taxonomy} for a more comprehensive treatment.}:
\begin{itemize}
    \item {Hoare Logic:} A forward over-approximation logic. A triple \psqhoare is a valid Hoare triple if and only if $Q \supseteq S(P)$.
    \item {Incorrectness Logic:} A forward under-approximation logic. A triple \psqincr is a valid Incorrectness triple if and only if $Q \subseteq S(P)$.
    \item {Lisbon Logic:} A backward \textbf{under}-approximation logic \cite{Ascari23:sil}. A triple \psqlisbon is a valid Lisbon triple if and only if $P \subseteq S^{-1}(Q)$.
    \item {Necessary Preconditions Logic:} A backward \textbf{over}-approximation logic \cite{CousotNecPre}. A triple \psqnecpre is a valid necessary precondition triple if and only if $P \supseteq S^{-1}(Q)$.
\end{itemize}

Next, we will formally describe these logic in a language similar to logic of actions. For the following definitions, assume $S \subseteq \Sigma \times \Sigma $ is the program transition relation where $\Sigma$ is the set of states of $S$.

\begin{definition}[Valid Hoare Triple]
    A Hoare triple \psqhoare is valid if and only if 
    \[
    \forall \sigma':\; (\exists \sigma:\; P(\sigma) \wedge S(\sigma, \sigma')) \implies Q(\sigma')
    \]
\end{definition}

\begin{corollary}
    \label{cor:hoaredef}
    A Hoare triple \psqhoare is valid if and only if
    \[
    \forall \sigma, \sigma':\; P(\sigma) \wedge S(\sigma, \sigma') \implies Q(\sigma')
    \]
\end{corollary}

\begin{definition}[Valid Incorrectness Triple]
    An Incorrectness triple \psqincr is a valid Incorrectness triple \textit{if and only if} 
    \[
    \forall \sigma'.\; Q(\sigma') \implies
    \exists \sigma.\; P(\sigma) \wedge S(\sigma,\sigma').
    \]
\end{definition}

\begin{definition}[Valid Lisbon Triple]
    A Lisbon triple \psqlisbon is a valid Lisbon triple \textit{if and only if} 
    \[
    \forall \sigma.\; P(\sigma) \implies \exists \sigma'. \; Q(\sigma') \wedge S(\sigma, \sigma')
    \]

\end{definition}

\begin{definition}[Valid Necessary Preconditions Triple]
    A Necessary Preconditions triple \psqnecpre is a valid Necessary Preconditions triple \textit{if and only if} 
    \[
    \forall \sigma.\; (\exists \sigma'. \; Q(\sigma') \wedge S(\sigma, \sigma')) \implies P(\sigma)
    \]
\end{definition}

\begin{corollary}
    \label{cor:necpredef}
    A Necessary Preconditions triple \psqnecpre is valid if and only if
    \[
    \forall \sigma, \sigma':\; Q(\sigma') \wedge S(\sigma, \sigma') \implies P(\sigma)
    \]
\end{corollary}

Corollary \ref{cor:hoaredef} and \ref{cor:necpredef} are useful because it shows us checking validity of a Hoare triple and a Necessary Preconditions triple lays in the decidable fragment of the FOL (first order logic) since there is no quantifier alternation provided that the program transformation predicate $S$ is also inside the decidable fragment.

The validity criteria of Section \ref{sec:validity-sim} single out programs whose transition relation behaves functionally in one direction or the other. We record these properties of a transition relation $S \subseteq \Sigma \times \Sigma$.

\begin{definition}[Functionality conditions]
A program with transition relation $S \subseteq \Sigma \times \Sigma$ is
\begin{itemize}
  \item \emph{total} if $\ \forall \sigma.\ \exists \sigma'.\ S(\sigma,\sigma')$;
  \item \emph{deterministic} if $\ \forall \sigma,\sigma_1',\sigma_2'.\ S(\sigma,\sigma_1') \wedge S(\sigma,\sigma_2') \implies \sigma_1' = \sigma_2'$;
  \item \emph{backward total} if $\ \forall \sigma'.\ \exists \sigma.\ S(\sigma,\sigma')$;
  \item \emph{backward deterministic} if $\ \forall \sigma_1,\sigma_2,\sigma'.\ S(\sigma_1,\sigma') \wedge S(\sigma_2,\sigma') \implies \sigma_1 = \sigma_2$.
\end{itemize}
A total, deterministic $S$ is the graph of a total function on $\Sigma$; a backward total, backward deterministic $S$ is one whose converse is a total function. These are exactly the conditions under which the existential witness in a forward (resp.\ backward) simulation is uniquely pinned, which is what Theorems~8 and~11 exploit.
\end{definition}

\subsection{Invariants}
\label{sec:invariants}
 
We take a \emph{safety property} of a program to be a set of states $J \subseteq \Sigma$ (equivalently, a predicate on states; we use the set and predicate views interchangeably). A program \emph{satisfies} $J$ when every state it can ever reach lies in $J$. Writing $\mathit{reach}(S)$ for the set of states reachable from $\mathit{Init}$ under finitely many steps of the transition relation $S$, this is the inclusion
\[
  \mathit{reach}(S) \;\subseteq\; J .
\]
A set $J$ with this property is called an \emph{invariant} of the program: an over-approximation of its reachable states~\cite{MannaPnueli95,CC77}. The strongest invariant is $\mathit{reach}(S)$ itself, and every weaker invariant is obtained by enlarging it. We deliberately phrase safety in terms of states rather than traces: it is the formulation that meshes with the program logics of Section~\ref{sec:functionalLogics}, whose pre- and post-conditions are themselves predicates on states.
 
Reachability quantifies over all executions, so checking $\mathit{reach}(S) \subseteq J$ directly is a global obligation. The standard device for reducing it to a local one is to require $J$ to be closed under the transition relation, which turns the global inclusion into a single step check~\cite{Floyd67,MannaPnueli95}. In the image notation of Section~\ref{sec:axiomatic-sim} this closure is exactly the validity of the Hoare triple $\{J\}\,S\,\{J\}$, since $\{J\}\,S\,\{J\}$ holds iff $S(J) \subseteq J$. This motivates the following notion.
 
\begin{definition}[Hoare invariant]
\label{def:hoare-invariant}
A set $I \subseteq \Sigma$ is a \emph{Hoare invariant} of a program with
transition relation $S$ and initial states $\mathit{Init}$ if
\begin{enumerate}
  \item $\mathit{Init} \subseteq I$ \hfill\textup{(initiation)}, and
  \item $\{I\}\,S\,\{I\}$ is a valid Hoare triple, i.e.\ $S(I) \subseteq I$
        \hfill\textup{(consecution)}.
\end{enumerate}
\end{definition}
 
Definition~\ref{def:hoare-invariant} is the classical notion of an \emph{inductive invariant}, whose two conditions are traditionally called \emph{initiation} and \emph{consecution}, going back to Floyd's inductive assertions~\cite{Floyd67} and presented in this transition system form by Manna and Pnueli~\cite{MannaPnueli95}; we phrase it through the Hoare triple of Section~\ref{sec:functionalLogics} to align it with the program logics used in the rest of the paper. Inductiveness makes a Hoare invariant a genuine invariant where a one step induction is enough, with no reasoning about whole executions.
 
\begin{proposition}
\label{prop:hoare-inv-is-inv}
    Every Hoare invariant $I$ is an invariant: $\mathit{reach}(S) \subseteq I$.
\end{proposition}
 
\begin{proof}
    By induction on the length of a reaching execution. Every state of length $0$ is in $\mathit{Init} \subseteq I$ by initiation. If a reachable state $s$ lies in $I$ and $S(s,s')$, then $s' \in S(I) \subseteq I$ by consecution. Hence $\mathit{reach}(S) \subseteq I$.
\end{proof}
 
The converse fails: an invariant need not be closed under $S$, and proving a desired safety property typically requires \emph{strengthening} it to a Hoare invariant that implies it~\cite{MannaPnueli95,CC77}. This is the role Hoare invariants play in our development. On the specification we establish a Hoare invariant $\{I\}\,S_2\,\{I\}$ with $\mathit{Init}_2 \subseteq I$; the simulation relation then transfers it to the implementation (Remark~\ref{rem:traces-to-invariants}), yielding a Hoare invariant on $S_1$ whose initiation and consecution are discharged by the transfer theorems of Section~\ref{sec:axiomatic-sim}. Thus a safety property proved once, on the simpler specification, is carried down to the implementation as a state predicate rather than as a property of traces.

%% file: arxivaxiomaticsim.tex
Relational Hoare logics (RHLs) provide axiomatic semantics for reasoning over multiple programs. In this work, we focus on logics relating two programs. The class of RHLs used in refinement proofs are called $\forall\exists$ RHLs since they are used for proving for all steps of the implementation program, there exists a related state of the specification program. 

For the scope this work, we want RHLs corresponding to the standard backward and forward simulation definitions of \cite{LynchV95:forward}.

We start by defining forward relational triple (FRT) and backward relational triple (BRT) in order to define forward and backward simulations in logic of actions. Let $\Sigma$ and $T$ be states of programs $S_1$ and $S_2$; and $P$ and $Q$ $\subseteq \Sigma \times T$ be binary predicates.

\begin{definition}[Forward Relational Triple]
\label{def:frt}
    Let $S_1$ and $S_2$ be two programs executing on states from $\Sigma$ and $T$, and $P$ and $Q \subseteq\Sigma \times T$ be a binary predicate. Then, $|P\rangle\ S_1; S_2\ |Q \rangle$ is a valid Forward Relational Triple if and only if
    \[
        \forall \sigma, \sigma', \tau. S_1(\sigma, \sigma') \wedge P(\sigma, \tau) \implies \exists \tau'. S_2(\tau, \tau') \wedge Q(\sigma', \tau')
    \]
\end{definition}

\begin{definition}[Backward Relational Triple]
\label{def:brt}
    Let $S_1$ and $S_2$ be two programs executing on states from $\Sigma$ and $T$, and $P$ and $Q \subseteq\Sigma \times T$ be a binary predicate. Then, $|P\rangle\ S_1; S_2\ |Q \rangle$ is a valid Backward Relational Triple if and only if
    \[
        \forall \sigma, \sigma', \tau'. S_1(\sigma, \sigma') \wedge P(\sigma', \tau') \implies \exists \tau. S_2(\tau, \tau') \wedge Q(\sigma, \tau)
    \]
\end{definition}

\begin{definition}[Forward Simulation Relation]
\label{def:fs}
    A relation $R \subseteq \Sigma \times T$ is a forward simulation relation if
    \[R[start(S_1)] \cap start(S_2) \neq \emptyset\]
    \[ \textit{for all atomic coordinated steps } C_1, C_2 \textit{ of } S_1, S_2; \textit {we have } |R\rangle\ C_1; C_2\ |R\rangle\]
\end{definition}

\begin{definition}[Backward Simulation Relation]
\label{def:bs}
    A relation $R \subseteq \Sigma \times T$ is a backward simulation relation if
    \[R[start(S_1)] \subseteq start(S_2)\]
    \[\textit{for all atomic coordinated steps } C_1, C_2 \textit{ of } S_1, S_2; \textit{ we have } \langle R|\ C_1; C_2\ \langle R|\]
\end{definition}




A standard refinement proof consists of two parts: (i) showing correctness of some properties on the specification program, (ii) showing the implementation refines the specification by finding a simulation relation between two programs. Then, one can infer some knowledge about the implementation.

In this work, we want to understand what kind of guarantees we can get for implementation programs if we use a functional program logic from Section \ref{sec:functionalLogics} for the first step and a $\forall\exists$ RHL from Section \ref{sec:axiomatic-sim} for the second step.

We claim that if a property represented by a Hoare invariant holds for the specification program and there is a forward simulation relation from the implementation to the specification, then the corresponding property imposed by the simulation relation holds for the implementation as a Hoare invariant, too. Similarly, an invariant that holds for the specification program, and there is a backward simulation relation from the implementation to the specification, then the corresponding property imposed by the simulation relation holds for the implementation as an invariant too. 

\begin{definition}[Forward Image along a Simulation Relation]
    Let $R(\sigma,\tau)$ relate implementation states $\sigma$ to specification states $\tau$.
    For a predicate $P$ on implementation states,
    \[
    R(P) \;\triangleq\;
    \{\;\tau \mid \exists \sigma.\; P(\sigma)\ \wedge\ R(\sigma,\tau)\;\}.
    \]
\end{definition}

\begin{definition}[Inverse Image along a Simulation Relation]
    Let $R(\sigma,\tau)$ be as above.
    For a predicate $Q$ on specification states,
    \[
    R^{-1}(Q) \;\triangleq\;
    \{\;\sigma \mid \exists \tau.\; Q(\tau)\ \wedge\ R(\sigma,\tau)\;\}.
    \]
\end{definition}

\begin{theorem}
\label{thm:hoareforward}
    Let I be a Hoare invariant on $S_2$, and $R(S_1,S_2)$ be a forward simulation relation. Then, $R^{-1}(I)$ is a Hoare invariant on $S_1$.
\end{theorem}
\begin{proof}
    We verify that $R^{-1}(I)$ satisfies initiation and consecution on $S_1$.
     
    \emph{Initiation.} Let $\sigma\in\mathit{Init}(S_1)$. By the start condition of Definition~10 there is $\tau\in\mathit{Init}(S_2)$ with $R(\sigma,\tau)$. Since $I$ is a Hoare invariant on $S_2$, initiation of $I$ gives $\mathit{Init}(S_2)\subseteq I$, so $I(\tau)$. By Definition~13, $\sigma\in R^{-1}(I)$. Hence $\mathit{Init}(S_1)\subseteq R^{-1}(I)$.
     
    \emph{Consecution.} Let $\sigma\in R^{-1}(I)$ and $S_1(\sigma,\sigma')$. By Definition~13 there is $\tau$ with $R(\sigma,\tau)$ and $I(\tau)$. The forward step condition (Definition~10) applied to $S_1(\sigma,\sigma')$ and $R(\sigma,\tau)$ yields $\tau'$ with $S_2(\tau,\tau')$ and $R(\sigma',\tau')$. Consecution of $I$, i.e. $S_2(I)\subseteq I$, gives $I(\tau')$ from $I(\tau)\wedge S_2(\tau,\tau')$. Thus $R(\sigma',\tau')$ and $I(\tau')$, so $\sigma'\in R^{-1}(I)$. Hence $S_1\!\big(R^{-1}(I)\big)\subseteq R^{-1}(I)$, i.e. $\{R^{-1}(I)\}\,S_1\,\{R^{-1}(I)\}$ is valid.
 
    Therefore $R^{-1}(I)$ is a Hoare invariant on $S_1$.
\end{proof}


\begin{lemma}
\label{lem:reachforward}
    Let $\sigma$ be a reachable state in $S_1$ and let $R(S_1,S_2)$ be a forward simulation. Then \(\exists \tau. R(\sigma,\tau) \wedge\) $\tau$ is reachable in $S_2$.
\end{lemma}
\begin{proof}
    By induction on the length $n$ of an execution of $S_1$ reaching $\sigma$.
     
    \emph{Base ($n=0$).} Then $\sigma\in\mathit{Init}(S_1)$. By the start condition there is $\tau\in\mathit{Init}(S_2)$ with $R(\sigma,\tau)$, and $\tau$ is reachable.
     
    \emph{Step.} Let $\sigma$ be reached by an execution $\cdots\sigma_0 \to\sigma$ with $\sigma_0$ reachable in $n-1$ steps and $S_1(\sigma_0,\sigma)$. By the induction hypothesis there is a reachable $\tau_0$ with $R(\sigma_0,\tau_0)$. The forward step condition applied to $S_1(\sigma_0,\sigma)$ and $R(\sigma_0,\tau_0)$ yields $\tau$ with $S_2(\tau_0,\tau)$ and $R(\sigma,\tau)$. Since $\tau_0$ is reachable and $S_2(\tau_0,\tau)$, $\tau$ is reachable.
\end{proof}

\begin{theorem}
\label{thm:invforward}
    Let $I$ be an invariant on $S_2$, and $R(S_1,S_2)$ be a forward simulation relation. Then $R^{-1}(I)$ is an invariant on $S_1$.
\end{theorem}
\begin{proof}
    Let $\sigma\in\mathit{reach}(S_1)$. By Lemma~1 there is $\tau$ with $R(\sigma,\tau)$ and $\tau\in\mathit{reach}(S_2)$. Since $I$ is an invariant on $S_2$, $\mathit{reach}(S_2)\subseteq I$, so $I(\tau)$. By Definition~13, $\sigma\in R^{-1}(I)$. Hence $\mathit{reach}(S_1)\subseteq R^{-1}(I)$, i.e. $R^{-1}(I)$ is an invariant on $S_1$.
\end{proof}

\begin{lemma}
\label{lem:reachbackward}
    Let $\sigma$ be a reachable state in $S_1$, and $R(S_1,S_2)$ be a backward simulation relation. Then, \(\forall \tau. R(\sigma, \tau) \implies \tau\) is reachable in $S_2$.
\end{lemma}
\begin{proof}
    By induction on the length $n$ of an execution of $S_1$ reaching $\sigma$.
     
    \emph{Base ($n=0$).} Then $\sigma\in\mathit{Init}(S_1)$. By the initiation condition of Definition~11, every $\tau$ with $R(\sigma,\tau)$ lies in $\mathit{Init}(S_2)$ and is therefore reachable.
     
    \emph{Step.} Let $\sigma$ be reached by $\cdots\sigma_0\to\sigma$ with $\sigma_0$ reachable and $S_1(\sigma_0,\sigma)$, and let $\tau$ be any state with $R(\sigma,\tau)$. The backward step condition applied to $S_1(\sigma_0,\sigma)$ and $R(\sigma,\tau)$ yields $\tau_0$ with $S_2(\tau_0,\tau)$ and $R(\sigma_0,\tau_0)$. By the induction hypothesis applied to the reachable $\sigma_0$ and $\tau_0$ (which satisfies $R(\sigma_0,\tau_0)$), $\tau_0$ is reachable; since $S_2(\tau_0,\tau)$, $\tau$ is reachable.
\end{proof}

\begin{theorem}
\label{thm:invbackward}
    Let I be an invariant on $S_2$, and $R(S_1,S_2)$ be a backward simulation relation. Then, $R^{-1}(I)$ is an invariant on $S_1$.
\end{theorem}
\begin{proof}
    Let $\sigma\in\mathit{reach}(S_1)$. By totality (Definition~11) there is $\tau$ with $R(\sigma,\tau)$. By Lemma~2, $\tau\in\mathit{reach}(S_2)$. Since $I$ is an invariant on $S_2$, $I(\tau)$. By Definition~13, $\sigma\in R^{-1}(I)$. Hence $\mathit{reach}(S_1)\subseteq R^{-1}(I)$, i.e. $R^{-1}(I)$ is an invariant on $S_1$.
\end{proof}

\begin{corollary}
\label{cor:hoarebackward}
    Let I be a Hoare invariant on $S_2$, and $R(S_1,S_2)$ be a backward simulation relation. Then, $R^{-1}(I)$ is an invariant on $S_1$.
\end{corollary}
\begin{proof}
    Every Hoare invariant is an invariant (Proposition~1), so $I$ is an invariant on $S_2$ and Theorem~7 gives that $R^{-1}(I)$ is an invariant on $S_1$.
\end{proof}

\subsection{Application: Functional Correctness of the Counters}
\label{sec:counter-application}
 
We apply the theory to the three counters of Section~\ref{sec:background}, whose states carry the history field $\mathit{old}$ and whose simulation relations contain the history containments (B2) and (F4). We transfer one property down the chain: that the history set records every value up to the current count. On each machine it has the shape $\forall v \le \textit{count}.\ v \in \mathit{old}$, with the count instantiated to $x$, $|C|$, and $L+R$ in turn. Established as a Hoare invariant on the sequential counter, it is carried to NDC by the backward relation and to LR by the forward relation, arriving as $\forall v \le L+R.\ v \in \mathit{old}_{lr}$ on the implementation.
 
We first record the contents of the history fields fixed by the update rules of Section~\ref{sec:background}. They are used to interpret the transferred invariants, not in the transfer itself.
 
\begin{proposition}[Structural invariants]\label{prop:struct}
    The augmented machines satisfy, at all reachable states,
    \[
      \mathit{old}_s = \{0,\dots,x\}, \qquad
      \mathit{old}_n = \{0,\dots,|C|+|P|\}, \qquad
      \mathit{old}_{lr} = \{0,\dots,L+R\}.
    \]
\end{proposition}
 
\begin{proof}
Each is a Hoare invariant by one step induction: initially every count is $0$ and $\mathit{old} = \{0\}$; each increment raises the relevant begun count ($x$, $|C|+|P|$, and $L+R$ respectively) by one and inserts exactly that new value into $\mathit{old}$; reads change neither.
\end{proof}
 
The two history containments are not independent obligations: under Proposition~\ref{prop:struct}, (B2) $\mathit{old}_s \subseteq \mathit{old}_n$ is the upper half of (B1), and (F4) $\mathit{old}_n \subseteq \mathit{old}_{lr}$ follows from (F1) once (F5) and (F6) force $P = \emptyset$ at related states; this is why adding the clauses left the simulation proofs of Section~\ref{sec:background} essentially unchanged.
 
\paragraph{The specification invariant.}
On the sequential counter we take
\[
  I_s \;\triangleq\; \forall v.\ v \le x \Rightarrow v \in \mathit{old}_s,
\]
``every value up to the current count has been recorded.''
 
\begin{proposition}\label{prop:Is}
    $I_s$ is a Hoare invariant on the sequential counter.
\end{proposition}
 
\begin{proof}
    \emph{Initiation.} Initially $x = 0$ and $\mathit{old}_s = \{0\}$, so the only $v \le 0$ is $v = 0 \in \mathit{old}_s$.
    
    \emph{Consecution.}
    \textsf{increment} sends $x$ to $x+1$ and $\mathit{old}_s$ to $\mathit{old}_s \cup \{x+1\}$; assuming $\forall v \le x.\ v \in \mathit{old}_s$, any $v \le x+1$ is either $\le x$, hence in $\mathit{old}_s$, or equal to $x+1 \in \mathit{old}_s \cup \{x+1\}$.
    \textsf{read} changes neither $x$ nor $\mathit{old}_s$.
\end{proof}
 
\paragraph{Step 1: Seq $\to$ NDC (backward).}
By Corollary~4, $R_B^{-1}(I_s)$ is an invariant on NDC. Define
\[
  I_n \;\triangleq\; \forall v.\ v \le |C| \Rightarrow v \in \mathit{old}_n.
\]
 
\begin{corollary}\label{cor:In}
    $I_n$ is an invariant on the NDC counter.
\end{corollary}
 
\begin{proof}
    Since a superset of an invariant is an invariant, it suffices to show $R_B^{-1}(I_s) \subseteq I_n$. Let $\sigma_n \in R_B^{-1}(I_s)$, witnessed by a Seq state $\tau_s$ with $R_B(\sigma_n,\tau_s)$ and $I_s(\tau_s)$, and let $v \le |C|$. By (B1), $|C| \le x$, so $v \le x$; by $I_s(\tau_s)$, $v \in \mathit{old}_s$; by the history clause (B2), $\mathit{old}_s \subseteq \mathit{old}_n$. Hence $v \in \mathit{old}_n$.
\end{proof}
 
The bound $|C|$ in $I_n$ is the \emph{lower} end of the interval $[|C|,|C|+|P|]$ of (B1): the inverse image existentially chooses the sequential count $x$, and the weakest, hence the only guaranteed, witness sits at $x = |C|$. In fact $\mathit{old}_n$ contains every value up to $|C|+|P|$ (Proposition~\ref{prop:struct}), but the backward transfer certifies only the range up to $|C|$; the pending range is not carried across. This is the one place in the chain where a relation weakens the transferred predicate, and it is why $I_n$ is stated as $R_B^{-1}(I_s)$ subsetted rather than as a count substituted copy of $I_s$.
 
\paragraph{Step 2: NDC $\to$ LR (forward).}
$I_n$ is an invariant -- not a Hoare invariant -- on NDC, so Theorem \ref{thm:hoareforward} does not apply and we use Theorem~6: $R_F^{-1}(I_n)$ is an invariant on LR. Define
\[
  I_{lr} \;\triangleq\; \forall v.\ v \le L+R \Rightarrow v \in \mathit{old}_{lr}.
\]
 
\begin{corollary}\label{cor:Ilr}
    $I_{lr}$ is an invariant on the LR counter.
\end{corollary}
 
\begin{proof}
    It suffices to show $R_F^{-1}(I_n) \subseteq I_{lr}$. Let $\sigma_{lr} \in R_F^{-1}(I_n)$, witnessed by an NDC state $\sigma_n$ with $R_F(\sigma_{lr},\sigma_n)$ and $I_n(\sigma_n)$, and let $v \le L+R$. By (F1), $L+R = |C|$, so $v \le |C|$; by $I_n(\sigma_n)$, $v \in \mathit{old}_n$; by the history clause (F4), $\mathit{old}_n \subseteq \mathit{old}_{lr}$. Hence $v \in \mathit{old}_{lr}$.
\end{proof}
 
Unlike Step~1, this step loses nothing: (F1) is an equality, $L+R = |C|$, so the bound transfers at full strength with no existential slack.
 
\paragraph{Summary.}
A single Hoare invariant on the trivially correct sequential counter (Proposition~\ref{prop:Is}) is transported, by Corollary~\ref{cor:In} and then Theorem~\ref{thm:invforward}, to the containment invariant $\forall v \le L+R.\ v \in \mathit{old}_{lr}$ on the LR implementation (Corollary~\ref{cor:Ilr}). The same containment is, of course, immediate from the LR update rule (Proposition~\ref{prop:struct}); the content of the result is that it is reached purely by transfer -- through the relation clauses (B1), (B2), (F1), (F4) -- with no reasoning about the LR machine's own dynamics. The backward leg weakens the predicate to the floor of the (B1) interval, while the forward leg, resting on the equality (F1), carries it across exactly. Finally, since an LR read reports $v_L(id)+R \le L+R$, the containment $I_{lr}$ gives $o_{lr}(id) \in \mathit{old}_{lr}$: every completed read returns a value the counter realized.

%% file: arxivvaliditysim.tex
We are using forward and backward simulations to relate the implementation to the specification. In order to verify that a relation is a forward simulation or a backward simulation, functional program logic can be used. For verifying that a relation is a forward simulation, Hoare triples can be used as in Theorem \ref{thm:forwardAsHoare} under the assumption that the specification program is total and deterministic. 

\begin{theorem}
\label{thm:forwardAsHoare}
    If $S_2$ is total and deterministic, then \(\hoare{R(\sigma,\tau)}{S_2;S_1;}{R(\sigma,\tau)}\) is a valid Hoare triple if and only if R is a forward simulation.
\end{theorem}
\begin{proof}
    Let $T$ be the product transition $T\big((\sigma,\tau),(\sigma',\tau')\big)\equiv S_1(\sigma,\sigma')\wedge S_2(\tau,\tau')$. By Corollary~1, the Hoare triple $\{R\}\,T\,\{R\}$ is valid iff
    \[ (\star)\quad \forall\sigma,\sigma',\tau,\tau'.\; R(\sigma,\tau)\wedge S_1(\sigma,\sigma')\wedge S_2(\tau,\tau') \;\Rightarrow\; R(\sigma',\tau'). \]
    The forward simulation condition (Definition~10) is
    \[ (\dagger)\quad \forall\sigma,\sigma',\tau.\; R(\sigma,\tau)\wedge S_1(\sigma,\sigma') \;\Rightarrow\; \exists\tau'.\; S_2(\tau,\tau')\wedge R(\sigma',\tau'). \]
     
    $(\dagger\Rightarrow\star)$ Assume $(\dagger)$ and the antecedent of $(\star)$. From $R(\sigma,\tau)\wedge S_1(\sigma,\sigma')$, $(\dagger)$ gives $\tau''$ with $S_2(\tau,\tau'')\wedge R(\sigma',\tau'')$. As $S_2$ is deterministic and $S_2(\tau,\tau')$ holds, $\tau''=\tau'$, hence $R(\sigma',\tau')$.
     
    $(\star\Rightarrow\dagger)$ Assume $(\star)$ and $R(\sigma,\tau)\wedge S_1(\sigma,\sigma')$. As $S_2$ is total there is $\tau'$ with $S_2(\tau,\tau')$. Then $(\star)$ gives $R(\sigma',\tau')$, so $\exists\tau'.\,S_2(\tau,\tau')\wedge R(\sigma',\tau')$.
\end{proof}

However, the specification being total and deterministic cannot always be guaranteed. In such a case, Lisbon triples can be used to verify that a relation is a forward simulation as in Theorem \ref{thm:forwardAsLisbon}, or alternatively as in Theorem \ref{thm:forwardAsLisbons}.

\begin{theorem}
\label{thm:forwardAsLisbon}
    $\lisbon{R(\sigma,\tau) \wedge S_1(\sigma,\sigma')}{S_2(\tau,\tau')}{R(\sigma',\tau')}$ is a valid Lisbon triple $\iff$ R is a forward simulation.
\end{theorem}
\begin{proof}
    The program $S_2$ transforms the specification state $\tau$ into $\tau'$, and $\sigma,\sigma'$ are free. By the definition of a valid Lisbon triple (Definition~7), $\langle R(\sigma,\tau)\wedge S_1(\sigma,\sigma')\rangle\, S_2\,\langle R(\sigma',\tau')\rangle$ is valid iff
    \[ \forall\tau.\; \big(R(\sigma,\tau)\wedge S_1(\sigma,\sigma')\big) \;\Rightarrow\; \exists\tau'.\; R(\sigma',\tau')\wedge S_2(\tau,\tau'). \]
    A triple with free variables is valid iff it holds for all their valuations; taking the conjunction over $\sigma,\sigma'$ gives
    \[ \forall\sigma,\sigma',\tau.\; R(\sigma,\tau)\wedge S_1(\sigma,\sigma') \;\Rightarrow\; \exists\tau'.\; S_2(\tau,\tau')\wedge R(\sigma',\tau'), \]
    which is exactly Definition~10. (When $S_1(\sigma,\sigma')$ is false the premise fails and the triple holds vacuously, matching the vacuous case of the simulation condition.)
\end{proof}

\begin{theorem}
\label{thm:forwardAsLisbons}
    Let \(R_1 = \{\sigma | \exists \tau. R(\sigma, \tau)\}\) and \(R_1(\sigma) = \{\tau | \sigma \in R_1 \wedge R(\sigma,\tau)\}\). Then, \\
        \((\forall \sigma, \sigma'. S_1(\sigma,\sigma') \implies \lisbon{R_1(\sigma)}{S_2}{R_1(\sigma')} \text{ is a valid Lisbon triple}) \iff \text{ R is a forward simulation.}\)
\end{theorem}
\begin{proof}
    For $\sigma\in R_1$ we have $R_1(\sigma)=\{\tau\mid R(\sigma,\tau)\}$, and for $\sigma\notin R_1$ we have $R_1(\sigma)=\emptyset=\{\tau\mid R(\sigma,\tau)\}$; in both cases $R_1(\sigma)(\tau)\iff R(\sigma,\tau)$, and likewise for $\sigma'$. By Definition~7, for fixed $\sigma,\sigma'$ the triple $\langle R_1(\sigma)\rangle\,S_2\,\langle R_1(\sigma')\rangle$ is valid iff
    \[ \forall\tau.\; R(\sigma,\tau)\;\Rightarrow\; \exists\tau'.\; R(\sigma',\tau')\wedge S_2(\tau,\tau'). \]
    Hence the stated condition $\forall\sigma,\sigma'.\,S_1(\sigma,\sigma')\Rightarrow(\text{this triple is valid})$ is
    \[ \forall\sigma,\sigma',\tau.\; S_1(\sigma,\sigma')\wedge R(\sigma,\tau) \;\Rightarrow\; \exists\tau'.\; S_2(\tau,\tau')\wedge R(\sigma',\tau'), \]
    which is the forward simulation condition (Definition~10).
\end{proof}

The relation defined in \ref{sec:fw-rel} is verified as a forward simulation relation using Theorem \ref{thm:forwardAsLisbons}. For each implementation operation, the corresponding Lisbon triple is shown to be valid using the proof rules defined by \cite{Ascari23:sil}. See Appendix \ref{app:verifyforward} for the full proof.

Similarly, to verify that a relation is a backward simulation, Necessary Precondition triples can be used as in Theorem \ref{thm:backwardAsNecpre} under the assumption that the specification is backward total and backward deterministic.

\begin{theorem}
\label{thm:backwardAsNecpre}
    If $S_2$ is backward total and backward deterministic, then \(\necpre{R(\sigma,\tau)}{S_2;S_1;}{R(\sigma,\tau)}\) is a valid Necessary Preconditions triple $\iff$ R is a backward simulation.
\end{theorem}
\begin{proof}
    Let $T$ be the product transition as in the proof of Theorem~8. By Corollary~2, the Necessary Preconditions triple $(R)\,T\,(R)$ is valid iff
    \[ (\star)\quad \forall\sigma,\sigma',\tau,\tau'.\; R(\sigma',\tau')\wedge S_1(\sigma,\sigma')\wedge S_2(\tau,\tau') \;\Rightarrow\; R(\sigma,\tau). \]
    The backward simulation condition (Definition~11) is
    \[ (\dagger)\quad \forall\sigma,\sigma',\tau'.\; S_1(\sigma,\sigma')\wedge R(\sigma',\tau') \;\Rightarrow\; \exists\tau.\; S_2(\tau,\tau')\wedge R(\sigma,\tau). \]
     
    $(\dagger\Rightarrow\star)$ Assume $(\dagger)$ and the antecedent of $(\star)$. From $S_1(\sigma,\sigma')\wedge R(\sigma',\tau')$, $(\dagger)$ gives $\tau''$ with $S_2(\tau'',\tau')\wedge R(\sigma,\tau'')$. As $S_2$ is backward deterministic and $S_2(\tau,\tau')$ holds, $\tau''=\tau$, hence $R(\sigma,\tau)$.
     
    $(\star\Rightarrow\dagger)$ Assume $(\star)$ and $S_1(\sigma,\sigma')\wedge R(\sigma',\tau')$. As $S_2$ is backward total there is $\tau$ with $S_2(\tau,\tau')$. Then $(\star)$ gives $R(\sigma,\tau)$, so $\exists\tau.\,S_2(\tau,\tau')\wedge R(\sigma,\tau)$.
\end{proof}

If the specification is not backward total or backward deterministic, Incorrectness triples can be used to verify that the relation is a backward simulation as in Theorem \ref{thm:backwardAsIncorrectness}, or alternatively as in Theorem \ref{thm:backwardAsIncorrectnesses}.

\begin{theorem}
\label{thm:backwardAsIncorrectness}
    $\incre{R(\sigma,\tau)}{S_2(\tau,\tau')}{R(\sigma',\tau') \wedge S_1(\sigma,\sigma')}$ is a valid Incorrectness triple $\iff$ R is a backward simulation.
\end{theorem}
\begin{proof}
    The program $S_2$ transforms $\tau$ into $\tau'$, and $\sigma,\sigma'$ are free. By the definition of a valid Incorrectness triple (Definition~6), $[R(\sigma,\tau)]\,S_2\,[R(\sigma',\tau')\wedge S_1(\sigma,\sigma')]$ is valid iff
    \[ \forall\tau'.\; \big(R(\sigma',\tau')\wedge S_1(\sigma,\sigma')\big) \;\Rightarrow\; \exists\tau.\; R(\sigma,\tau)\wedge S_2(\tau,\tau'). \]
    Holding for all valuations of the free $\sigma,\sigma'$, this is
    \[ \forall\sigma,\sigma',\tau'.\; R(\sigma',\tau')\wedge S_1(\sigma,\sigma') \;\Rightarrow\; \exists\tau.\; S_2(\tau,\tau')\wedge R(\sigma,\tau), \]
    which is the backward simulation condition (Definition~11).
\end{proof}

\begin{theorem}
\label{thm:backwardAsIncorrectnesses}
    Let \(R_1 = \{\sigma | \exists \tau. R(\sigma, \tau)\}\) and \(R_1(\sigma) = \{\tau | \sigma \in R_1 \wedge R(\sigma,\tau)\}\). Then, \\
        \((\forall \sigma, \sigma'. S_1(\sigma,\sigma') \implies \incre{R_1(\sigma)}{S_2}{R_1(\sigma')} \text{ is a valid Incorrectness triple}) \iff \text{ R is a backward simulation.}\)
\end{theorem}
\begin{proof}
    As in the proof of Theorem~10, $R_1(\sigma)(\tau)\iff R(\sigma,\tau)$ and $R_1(\sigma')(\tau')\iff R(\sigma',\tau')$. By Definition~6, for fixed $\sigma,\sigma'$ the triple $[R_1(\sigma)]\,S_2\,[R_1(\sigma')]$ is valid iff
    \[ \forall\tau'.\; R(\sigma',\tau')\;\Rightarrow\; \exists\tau.\; R(\sigma,\tau)\wedge S_2(\tau,\tau'). \]
    Hence $\forall\sigma,\sigma'.\,S_1(\sigma,\sigma')\Rightarrow(\text{this triple is valid})$ is
    \[ \forall\sigma,\sigma',\tau'.\; S_1(\sigma,\sigma')\wedge R(\sigma',\tau') \;\Rightarrow\; \exists\tau.\; S_2(\tau,\tau')\wedge R(\sigma,\tau), \]
    the backward simulation condition (Definition~11).
\end{proof}

The relation defined in \ref{sec:bw-rel} is verified as a backward simulation relation using Theorem \ref{thm:backwardAsIncorrectnesses}. For each implementation operation, the corresponding Incorrectness triple is shown to be valid using the proof rules defined by \cite{OHearn2019:incorrectness}. See Appendix \ref{app:verifybackward} for the full proof.

%% file: arxivconclusions.tex
We have related axiomatic functional verification logic to refinement proofs through simulation relations. On the transfer side, a forward simulation carries a Hoare invariant of the specification to a Hoare invariant of the implementation, and forward and backward simulations carry ordinary invariants, so a safety property proved once on a simple specification descends to the implementation as a state predicate. On the discharge side, we characterized forward simulations by the validity of Hoare and Lisbon triples and backward simulations by the validity of Necessary Preconditions and Incorrectness triples, reducing the simulation obligation to a triple in a standard functional logic.
 
We illustrated the framework on a concurrent counter. Linking a Left--Right implementation to an atomic sequential specification through an intermediate nondeterministic concurrent counter, a forward simulation from LR counter to NDC counter, and a backward simulation from the NDC counter to the sequential counter, we transferred the counter's safety bound down the chain without reasoning directly about the implementation's interleavings. The two simulations were verified operation by operation using Lisbon and Incorrectness triples (Appendices~\ref{app:verifyforward} and~\ref{app:verifybackward}).
 
We aim to extend the work in two directions. The first concerns the reach of the transfer results. We have shown that forward simulations preserve Hoare invariants and that backward simulations preserve invariants obtained by backward reasoning; it remains open which properties expressible in Incorrectness logic, Lisbon logic, and the Necessary Preconditions fragments are preserved under forward and backward simulations, and under what structural conditions. A systematic map from (logic, simulation direction) pairs to the class of preserved properties is the main goal. The second is to test the discharge method on implementations whose linearization points are genuinely value dependent, richer concurrent data structures than the counter, to see how far the operation by operation triple verification scales. Mechanizing the transfer theorems and the worked example in a proof assistant such as Lean~4 would additionally provide machine checked guarantees.

%% file: arxivappendix.tex
\section{Verification of Simulation Relations}
\label{app:proofs}

\subsection{Forward Simulation Verification}
\label{app:verifyforward}

By Theorem \ref{thm:forwardAsLisbons}, it suffices to show that for every implementation step $(\sigma, \sigma')$ of $\mathit{Ai}$ and every $\tau$ with $R(\sigma, \tau)$, the Lisbon triple $\langle R(\sigma, \tau) \rangle\ c_1;\ldots;c_n\ \langle R(\sigma', \tau') \rangle$ is valid for the corresponding sequence of specification steps. We derive each triple with the SIL proof rules. The following lemma is used in the increment case.

\begin{lemma}
\label{lem:noP}
    If $R(\sigma,\tau)$ then $P=\emptyset$.
\end{lemma}
\begin{proof}
    If $id\in P$ then $\mathit{pc}_n(id)=\textsc{done}\neq\textsc{done}$, so (F5) gives $\mathit{pc}_{lr}(id)\neq\textsc{done}$; as LR increments are atomic, $\mathit{pc}_{lr}(id)=\textsc{done}$, and then F~6 forces $\mathit{pc}_n(id)=\textsc{done}$, contradicting $\textsc{done}$.
\end{proof}

\noindent\textbf{Case 1:} $\textsf{increment}(\mathit{id})$.
 
\smallskip
\noindent\emph{Implementation step:}

Assumes $pc_{lr}(\mathit{id}) = \textsc{init}$; nondeterministically sets either $L \leftarrow L{+}1$ or $R \leftarrow R{+}1$; inserts the new $L+R$ into $old_{lr}$; sets $pc_{lr}(\mathit{id}) \leftarrow \textsc{done}$. Without loss of generality take $L \leftarrow L{+}1$ (the case $R \leftarrow R{+}1$ is symmetric).
 
\noindent\emph{Corresponding specification steps:}

$c_1 = \textsf{begin\_increment}(\mathit{id})$;
$c_2 = \textsf{return\_increment}(\mathit{id})$.
 
\noindent\emph{Precondition and postcondition:}
\begin{align*}
P &= ((|C| = L + R) \\
&\quad\;\land\; (\forall id \in \mathbb{I}. \ pc_{lr}(id) = \textsc{mid} \iff id \in A) \\
&\quad\;\land\; (\forall id \in A. \ |MS(id)| \le v_L(id) + R \le |MS(id) \cup YS(id)|) \\
&\quad\;\land\; (old_n \subseteq old_{lr}) \\
&\quad\;\land\; (\forall id \in \mathbb{I}. \ pc_n(id) \neq INIT \implies pc_{lr}(id) \neq INIT \wedge i_{lr}(id) = i_n(id)) \\
&\quad\;\land\; (\forall id \in \mathbb{I}. \ pc_{lr}(id) = DONE \implies pc_n(id) = DONE \wedge o_{lr}(id) = o_n(id))
)\\[2pt]
Q &= ((|C'| = L{+}R{+}1) \\
  &\quad\;\land\; (\forall \mathit{id}' \in I.\ pc_{lr}'(\mathit{id}') = \textsc{mid} \iff \mathit{id}' \in A')\\
  &\quad\;\land\; (\forall \mathit{id}' \in A'.\ |MS'(\mathit{id}')| \leq v_L(\mathit{id}') + R \leq |MS'(\mathit{id}') \cup YS'(\mathit{id}')|)\\
  &\quad\;\land\; (old_n' \subseteq old_{lr}')\\
  &\quad\;\land\; (\forall \mathit{id}' \in I.\ pc_n'(\mathit{id}') \neq \textsc{init} \implies pc_{lr}'(\mathit{id}') \neq \textsc{init} \land i_{lr}'(\mathit{id}') = \mathit{i}_n'(\mathit{id}'))\\
  &\quad\;\land\; (\forall \mathit{id}' \in I.\ pc_{lr}'(\mathit{id}') = \textsc{done} \implies pc_n'(\mathit{id}') = \textsc{done} \land o_{lr}'(\mathit{id}') = \mathit{o}_n'(\mathit{id}'))
\end{align*}
 
\noindent\emph{Proof of Lisbon triple $\langle P \rangle\ c_1; c_2\ \langle Q \rangle$.}
 
\noindent Define the intermediate condition:
\begin{align*}
Q_1 &= (|C \cup \{\mathit{id}\}| = L{+}R{+}1) \\
  &\quad\;\land\; (\forall \mathit{id}' \in I.\ pc_{lr}(\mathit{id}') = \textsc{mid} \iff \mathit{id}' \in A)\\
  &\quad\;\land\; (\forall \mathit{id}' \in A.\ |MS(\mathit{id}')| \leq v_L(\mathit{id}') + R \leq |MS(\mathit{id}') \cup YS(\mathit{id}')|)\\
  &\quad\;\land\; (old_n \subseteq old_{lr}') \\
  &\quad\;\land\; (\forall \mathit{id}' \in I.\ pc_n(\mathit{id}') \neq \textsc{init} \implies pc_{lr}(\mathit{id}') \neq \textsc{init} \land i_{lr}(\mathit{id}') = \mathit{i}_n(\mathit{id}'))\\
  &\quad\;\land\; (\forall \mathit{id}' \in I.\ pc_{lr}(\mathit{id'}) = \textsc{done} \implies pc_n(\mathit{id'}) = \textsc{done} \land o_{lr}(\mathit{id'}) = \mathit{o}_n(\mathit{id'}))
\end{align*}
and the pre-intermediate condition:
\begin{align*}
Q_0 &= (|C \cup \{\mathit{id}\}| = L{+}R{+}1) \\
  &\quad\;\land\; (\forall \mathit{id}' \in I.\ pc_{lr}(\mathit{id}') = \textsc{mid} \iff \mathit{id}' \in A)\\
  &\quad\;\land\; (\forall \mathit{id}' \in A.\ |MS(\mathit{id}')| \leq v_L(\mathit{id}') + R \leq |MS(\mathit{id}') \cup YS(\mathit{id}') \cup \{\mathit{id}\}|)\\
  &\quad\;\land\; (old_n \subseteq old_{lr}') \\
  &\quad\;\land\; (\forall \mathit{id}' \in I.\ pc_n(\mathit{id}') \neq \textsc{init} \implies pc_{lr}(\mathit{id}') \neq \textsc{init} \land i_{lr}(\mathit{id}') = \mathit{i}_n(\mathit{id}'))\\
  &\quad\;\land\; (\forall \mathit{id}' \in I.\ pc_{lr}(\mathit{id'}) = \textsc{done} \implies pc_n(\mathit{id'}) = \textsc{done} \land o_{lr}(\mathit{id'}) = \mathit{o}_n(\mathit{id'}))
\end{align*}
 
By the \emph{atom} rule, $\langle Q_1 \rangle\ c_2\ \langle Q \rangle$ is valid: $c_2$ moves $\mathit{id}$ from $P$ to $C$ and sets $pc_n(\mathit{id}) \leftarrow \textsc{done}$, so $|C \cup \{\mathit{id}\}|$ becomes $|C'|$, (F2) is preserved since $pc_{lr}(\mathit{id})$ moved to $\textsc{done}$, (F3) and (F4) are unchanged, and F~6 is satisfied by the $pc_n$ update.
 
By the \emph{atom} rule, $\langle Q_0 \rangle\ c_1\ \langle Q_1 \rangle$ is valid: $c_1$ adds $\mathit{id}$ to $P$ (so $|C \cup \{\mathit{id}\}| = L{+}R{+}1$ by (F1)) and updates $YS(t) \leftarrow YS(t) \cup \{\mathit{id}\}$ for all $t \in A$, which relaxes the upper bound in (F3) from $Q_0$ to $Q_1$, adds $|C \cup \{\mathit{id}\}|$ (since $P=\emptyset$ before $c_1$) to $old_n$ where (F4) is satisfied since $old_{lr}'$ also contains $L+R+1$, and sets $pc_n(\mathit{id}) \leftarrow \textsc{mid}$ satisfying (F5) since $pc_{lr}(\mathit{id}) = \textsc{done} \neq \textsc{init}$.
 
By the \emph{seq} rule, $\langle Q_0 \rangle\ c_1; c_2\ \langle Q \rangle$ is valid.
 
Finally, $P \subseteq Q_0$: From (F1) in $P$, $|C|=L{+}R$, and $id\notin C$ (else $\mathit{pc}_n(id)=\textsc{done}$ would give $\mathit{pc}_{lr}(id)\neq \textsc{init}$ by (F5)), so $|C\cup\{id\}|=L{+}R{+}1$. (F2) and (F5) agree between $P$ and $Q_0$ once $\mathit{pc}_{lr}(id)$ is read as $\textsc{done}$ ($id$ is neither $\textsc{mid}$ nor in $A$ on either side). The F3 bound of $Q_0$ is weaker than that of $P$ (its upper bound has the extra element $id$). For (F4), $\mathit{old}_n\subseteq\mathit{old}_{lr}\subseteq \mathit{old}_{lr}'$. The $id'\neq id$ part of F6 is inherited from $P$. Hence $P\subseteq Q_0$. By the \emph{cons} rule, $\langle P \rangle\ c_1; c_2\ \langle Q \rangle$ is valid.

\medskip
\noindent\textbf{Case 2:} $\textsf{begin\_read}(\mathit{id})$.
 
\smallskip
\noindent\emph{Implementation step:}

Assumes $pc_{lr}(\mathit{id}) = \textsc{init}$; sets $v_L(\mathit{id}) \leftarrow L$ and $pc_{lr}(\mathit{id}) \leftarrow \textsc{mid}$.
 
\noindent\emph{Corresponding specification steps:}

$c_1 = \textsf{begin\_read}(\mathit{id})$, which sets $A \leftarrow A \cup \{\mathit{id}\}$, $MS(\mathit{id}) \leftarrow C$, $YS(\mathit{id}) \leftarrow P$, and $pc_n(\mathit{id}) \leftarrow \textsc{mid}$.
 
\noindent\emph{Precondition and postcondition:}
\begin{align*}
P &= (|C| = L + R)\\
  &\quad\;\land\; (\forall \mathit{id}' \in I.\ pc_{lr}(\mathit{id}') = \textsc{mid} \iff \mathit{id}' \in A)\\
  &\quad\;\land\; (\forall \mathit{id}' \in A.\ |MS(\mathit{id}')| \leq v_L(\mathit{id}') + R \leq |MS(\mathit{id}') \cup YS(\mathit{id}')|)\\
  &\quad\;\land\; (old_n \subseteq old_{lr})\\
  &\quad\;\land\; (\forall \mathit{id}' \in I.\ pc_n(\mathit{id}') \neq \textsc{init} \implies pc_{lr}(\mathit{id}') \neq \textsc{init} \land i_{lr}(\mathit{id}') = \mathit{i}_n(\mathit{id}'))\\
  &\quad\;\land\; (\forall \mathit{id}' \in I.\ pc_{lr}(\mathit{id}') = \textsc{done} \implies pc_n(\mathit{id}') = \textsc{done} \land o_{lr}(\mathit{id}') = \mathit{o}_n(\mathit{id}'))\\[4pt]
Q &= (|C| = L + R)\\
  &\quad\;\land\; (\forall \mathit{id}' \in I.\ pc_{lr}'(\mathit{id}') = \textsc{mid} \iff \mathit{id}' \in A \cup \{\mathit{id}\})\\
  &\quad\;\land\; (\forall \mathit{id}' \in A \cup \{\mathit{id}\}.\ |MS'(\mathit{id}')| \leq v_L'(\mathit{id}') + R \leq |MS'(\mathit{id}') \cup YS'(\mathit{id}')|)\\
  &\quad\;\land\; (old_n \subseteq old_{lr})\\
  &\quad\;\land\; (\forall \mathit{id}' \in I.\ pc_n'(\mathit{id}') \neq \textsc{init} \implies pc_{lr}'(\mathit{id}') \neq \textsc{init} \land i_{lr}'(\mathit{id}') = \mathit{i}_n'(\mathit{id}'))\\
  &\quad\;\land\; (\forall \mathit{id}' \in I.\ pc_{lr}'(\mathit{id}') = \textsc{done} \implies pc_n'(\mathit{id}') = \textsc{done} \land o_{lr}'(\mathit{id}') = \mathit{o}_n'(\mathit{id}'))
\end{align*}
where $MS'(\mathit{id}) = C$, $YS'(\mathit{id}) = P$, $v_L'(\mathit{id}) = L$, and primed variables for $\mathit{id}' \neq \mathit{id}$ are unchanged from the pre-state.
 
\noindent\emph{Proof of Lisbon triple $\langle P \rangle\ c_1\ \langle Q \rangle$.} By the \emph{atom} rule it suffices to verify each invariant is established.
\begin{itemize}
  \item \textbf{(F1)}: $L$ and $R$ are unchanged, and $C$ is unchanged, so $|C'| = |C| = L + R$.
  \item \textbf{(F2)}: $pc_{lr}(\mathit{id})$ moves from $\textsc{init}$ to $\textsc{mid}$, and $\mathit{id}$ is added to $A$. For all other $\mathit{id}' \neq \mathit{id}$, both sides are unchanged.
  \item \textbf{(F3)}: For the new entry $\mathit{id} \in A'$: $MS(\mathit{id}) = C$ and $YS(\mathit{id}) = P$, so we need $|C| \leq v_L(\mathit{id}) + R \leq |C \cup P|$. Since $v_L(\mathit{id}) \leftarrow L$ and by (F1), $|C| = L + R$, we get $v_L(\mathit{id}) + R = L + R = |C|$, so both bounds hold with equality on the left. For existing $\mathit{id}' \in A$, $MS$, $YS$, $v_L$, and $R$ are all unchanged, so (F3) is preserved.
  \item \textbf{(F4)}: Neither $\mathit{old}_n$ nor $\mathit{old}_{lr}$ changes, so $\mathit{old}_n\subseteq\mathit{old}_{lr}$ is preserved.
  \item \textbf{(F5)}: $pc_n(\mathit{id})$ moves to $\textsc{mid} \neq \textsc{init}$, requiring $pc_{lr}(\mathit{id}) \neq \textsc{init}$, which holds since $pc_{lr}(\mathit{id})$ moved to $\textsc{mid}$. Inputs are null for reads so $i_{lr}(\mathit{id}) = \mathit{i}_n(\mathit{id}) = \mathtt{null}$. All other IDs are unaffected.
  \item \textbf{F~6}: No ID reaches $\textsc{done}$ in this step, so F~6 is vacuously preserved.
\end{itemize}
Hence $\langle P \rangle\ c_1\ \langle Q \rangle$ is valid by the \emph{atom} rule.

\medskip
\noindent\textbf{Case 3:} $\textsf{return\_read}(\mathit{id})$.
 
\smallskip
\noindent\emph{Implementation step:}

Assumes $pc_{lr}(\mathit{id}) = \textsc{mid}$; sets $pc_{lr}(\mathit{id}) \leftarrow \textsc{done}$ and $o_{lr}(\mathit{id}) \leftarrow v_L(\mathit{id}) + R$.
 
\noindent\emph{Corresponding specification steps}

$c_1 = \textsf{return\_read}(\mathit{id})$, which removes $\mathit{id}$ from $A$, nondeterministically chooses $S$ with $MS(\mathit{id}) \subseteq S \subseteq MS(\mathit{id}) \cup YS(\mathit{id})$, sets $\mathit{o}_n(\mathit{id}) \leftarrow |S|$ and $pc_n(\mathit{id}) \leftarrow \textsc{done}$.
 
\noindent\emph{Precondition and postcondition:}
\begin{align*}
P &= (|C| = L + R)\\
  &\quad\;\land\; (\forall \mathit{id}' \in I.\ pc_{lr}(\mathit{id}') = \textsc{mid} \iff \mathit{id}' \in A)\\
  &\quad\;\land\; (\forall \mathit{id}' \in A.\ |MS(\mathit{id}')| \leq v_L(\mathit{id}') + R \leq |MS(\mathit{id}') \cup YS(\mathit{id}')|)\\
  &\quad\;\land\; (old_n \subseteq old_{lr})\\
  &\quad\;\land\; (\forall \mathit{id}' \in I.\ pc_n(\mathit{id}') \neq \textsc{init} \implies pc_{lr}(\mathit{id}') \neq \textsc{init} \land i_{lr}(\mathit{id}') = \mathit{i}_n(\mathit{id}'))\\
  &\quad\;\land\; (\forall \mathit{id}' \in I.\ pc_{lr}(\mathit{id}') = \textsc{done} \implies pc_n(\mathit{id}') = \textsc{done} \land o_{lr}(\mathit{id}') = \mathit{o}_n(\mathit{id}'))\\[4pt]
Q &= (|C| = L + R)\\
  &\quad\;\land\; (\forall \mathit{id}' \in I.\ pc_{lr}'(\mathit{id}') = \textsc{mid} \iff \mathit{id}' \in A \setminus \{\mathit{id}\})\\
  &\quad\;\land\; (\forall \mathit{id}' \in A \setminus \{\mathit{id}\}.\ |MS(\mathit{id}')| \leq v_L(\mathit{id}') + R \leq |MS(\mathit{id}') \cup YS(\mathit{id}')|)\\
  &\quad\;\land\; (old_n \subseteq old_{lr})\\
  &\quad\;\land\; (\forall \mathit{id}' \in I.\ pc_n'(\mathit{id}') \neq \textsc{init} \implies pc_{lr}'(\mathit{id}') \neq \textsc{init} \land i_{lr}'(\mathit{id}') = \mathit{i}_n'(\mathit{id}'))\\
  &\quad\;\land\; (\forall \mathit{id}' \in I.\ pc_{lr}'(\mathit{id}') = \textsc{done} \implies pc_n'(\mathit{id}') = \textsc{done} \land o_{lr}'(\mathit{id}') = \mathit{o}_n'(\mathit{id}'))
\end{align*}
where $pc_{lr}'(\mathit{id}) = \textsc{done}$, $o_{lr}'(\mathit{id}) = v_L(\mathit{id}) + R$, $pc_n'(\mathit{id}) = \textsc{done}$, $\mathit{o}_n'(\mathit{id}) = |S|$, and all other variables are unchanged from the pre-state.
 
\noindent\emph{Choice of $S$.} We choose $S$ such that $|S| = v_L(\mathit{id}) + R$ to match the concrete return value and satisfy F~6. By (F3) in the pre-state, $|MS(\mathit{id})| \leq v_L(\mathit{id}) + R \leq |MS(\mathit{id}) \cup YS(\mathit{id})|$, so such an $S$ exists.
 
\noindent\emph{Proof of Lisbon triple $\langle P \rangle\ c_1\ \langle Q \rangle$.} By the \emph{atom} rule we verify each invariant.
\begin{itemize}
  \item \textbf{(F1)}: $L$, $R$, and $C$ are unchanged.
  \item \textbf{(F3)}: $\mathit{id}$ is removed from $A$, so the (F3) obligation for $\mathit{id}$ is dropped. For all remaining $\mathit{id}' \in A' = A \setminus \{\mathit{id}\}$, $MS$, $YS$, $v_L$, and $R$ are unchanged, so (F3) is preserved.
  \item \textbf{(F4)}: Neither $\mathit{old}_n$ nor $\mathit{old}_{lr}$ changes, so $\mathit{old}_n\subseteq\mathit{old}_{lr}$ is preserved.
  \item \textbf{(F5)}: $pc_n(\mathit{id})$ moves to $\textsc{done} \neq \textsc{init}$, requiring $pc_{lr}(\mathit{id}) \neq \textsc{init}$, which holds since $pc_{lr}(\mathit{id}) = \textsc{mid}$ in the pre-state. All other IDs are unaffected.
  \item \textbf{F~6}: $pc_{lr}(\mathit{id})$ moves to $\textsc{done}$, requiring $pc_n(\mathit{id}) = \textsc{done}$ and $o_{lr}(\mathit{id}) = \mathit{o}_n(\mathit{id})$. The specification step sets $pc_n(\mathit{id}) \leftarrow \textsc{done}$ and $\mathit{o}_n(\mathit{id}) \leftarrow |S| = v_L(\mathit{id}) + R = o_{lr}(\mathit{id})$, so both conditions hold by our choice of $S$. All other IDs are unaffected.
\end{itemize}
Hence $\langle P \rangle\ c_1\ \langle Q \rangle$ is valid by the \emph{atom}
rule.
 
\medskip
Since the Lisbon triple $\langle R(\sigma, \tau) \rangle\ c_1;\ldots;c_n\ \langle R(\sigma', \tau') \rangle$ is valid for every implementation operation, $R$ is a valid forward simulation by Theorem \ref{thm:forwardAsLisbons}.

\clearpage

\subsection{Backward Simulation Verification}
\label{app:verifybackward}

By Theorem \ref{thm:backwardAsIncorrectnesses}, with NDC as $S_1$ and Seq as $S_2$, it suffices to show that for every implementation step $(\sigma,\sigma')$ of NDC the Incorrectness triple $\incre{R_1(\sigma)}{c}{R_1(\sigma')}$ is valid for the corresponding sequential program $c$. The NDC state is a parameter (at $\sigma$ in the precondition $P$, at $\sigma'$ in the postcondition $Q$); the sequential state is the program state. We write $c=|C|$, $p=|P|$ for the pre-state NDC counts and reason over reachable states, using:
 
\begin{lemma}[Structural invariants]\label{lem:bstruct}
    At every reachable related state: \\
    \textup{(i)} $\mathit{old}_s=\{0,\dots,x\}$, $\mathit{old}_n=\{0,\dots, |C|+|P|\}$; \\
    \textup{(ii)} $x=|C|+|\{a\in P:\mathit{pc}_s(a)=\textsc{done}\}|$; \\
    \textup{(iii)} if $\mathit{pc}_s(id)=\textsc{done}$ then $o_s(id)$ is the value of $x$ at the step where $id$ linearized.
\end{lemma}
 
\paragraph{Case 1:} $\textsf{return\_increment}(a)$.

\emph{Implementation step:}

Assumes $\mathit{pc}_n(a)=\textsc{mid}$, $a\in P$; moves $a$ from $P$ to $C$; sets $\mathit{pc}_n(a)\leftarrow \textsc{done}$.

\emph{Corresponding specification step:}

$c=\mathtt{skip}$.

\begin{align*}
     P &= (|C|\le x\le|C|+|P|) \\
     &\quad\;\land\; (\mathit{old}_s\subseteq\mathit{old}_n) \\
     &\quad\;\land\; (\forall id.\ \mathit{pc}_s(id)\neq\textsc{init}\implies \mathit{pc}_n(id)\neq\textsc{init}\wedge i_n(id)=i_s(id)) \\
     &\quad\;\land\; (\forall id.\ \mathit{pc}_n(id)=\textsc{done}\implies \mathit{pc}_s(id)=\textsc{done}\wedge o_n(id)=o_s(id)) \\
     Q &= (|C|{+}1\le x\le|C|+|P|) \\
     &\quad\;\land\; (\mathit{old}_s\subseteq\mathit{old}_n)\\
     &\quad\;\land\; (\forall id.\ \mathit{pc}_s(id)\neq\textsc{init}\implies \mathit{pc}_n'(id)\neq\textsc{init}\wedge i_n(id)=i_s(id))\\
     &\quad\;\land\; (\forall id.\ \mathit{pc}_n'(id)=\textsc{done}\implies \mathit{pc}_s(id)=\textsc{done}\wedge o_n(id)=o_s(id)),
\end{align*}

where $\mathit{pc}_n'(a)=\textsc{done}$ and $|C'|=|C|{+}1$, $|P'|=|P|{-}1$ (so $Q$'s (B1) reads $|C|{+}1\le x\le|C|+|P|$).
 
Since $\incre{P}{skip}{P}$ is a valid Incorrectness triple due to the \emph{unit} rule, by the \emph{consequence} rule it suffices to show that $Q\implies P$. Let's inspect this clause by clause:
\begin{itemize}
    \item \textbf{(B1):} $|C|{+}1 \le x \le |C|{+}|P| \implies |C| \le x \le |C|{+}|P|$.
    \item \textbf{(B2):} $\mathit{old}_n$ is unchanged, identical.
    \item \textbf{(B3):} $\mathit{pc}_n'(a)=\textsc{done}\neq\textsc{init}$ and $\mathit{pc}_n(a)=\textsc{mid}\neq\textsc{init}$, so $Q$'s clause implies $P$'s as they are both vacuously true; other ids are unchanged.
    \item \textbf{(B4):} $\mathit{pc}_n(a)=\textsc{mid}$ makes $P$'s clause vacuous for $a$, while $Q$ constrains it, so $Q$'s clause implies $P$'s clause; other ids are unchanged.
\end{itemize}
Hence $Q\implies P$ and $\incre{P}{skip}{Q}$ is a valid Incorrectness triple.
 
\paragraph{Case 2:} $\textsf{return\_read}(a)$.

\emph{Implementation step:}

Assumes $\mathit{pc}_n(a)=\textsc{mid}$, $a\in A$; removes $a$ from $A$; chooses $S$ with $\mathit{MS}(a) \subseteq S\subseteq\mathit{MS}(a)\cup\mathit{YS}(a)$; sets $o_n(a) \leftarrow|S|$, $\mathit{pc}_n(a)\leftarrow\textsc{done}$.
\emph{Corresponding specification step:}

$c=\mathtt{skip}$.

\begin{align*}
    P &= (|C|\le x\le|C|+|P|) \\
    &\quad\;\land\; (\mathit{old}_s\subseteq\mathit{old}_n) \\ 
    &\quad\;\land\; (\forall id.\ \mathit{pc}_s(id)\neq\textsc{init}\implies \mathit{pc}_n(id)\neq\textsc{init}\wedge i_n(id')=i_s(id')) \\
    &\quad\;\land\; (\forall id.\ \mathit{pc}_n(id)=\textsc{done}\implies \mathit{pc}_s(id)=\textsc{done}\wedge o_n(id)=o_s(id)) \\
    Q &= (|C|\le x\le|C|+|P|) \\
    &\quad\;\land\; (\mathit{old}_s\subseteq\mathit{old}_n) \\
    &\quad\;\land\; (\forall id.\ \mathit{pc}_s(id)\neq\textsc{init}\implies \mathit{pc}_n'(id)\neq\textsc{init}\wedge i_n(id)=i_s(id)) \\
    &\quad\;\land\; (\forall id.\ \mathit{pc}_n'(id)=\textsc{done}\implies \mathit{pc}_s(id)=\textsc{done}\wedge o_n'(id)=o_s(id))
\end{align*}
where $\mathit{pc}_n'(a)=\textsc{done}$, $o_n'(a)=|S|$, and $|C|,|P|$ are unchanged.
 
Since $\incre{P}{skip}{P}$ is a valid Incorrectness triple due to the \emph{unit} rule, by the \emph{consequence} rule it suffices to show that $Q\implies P$. Let's inspect this clause by clause:
\begin{itemize}
    \item \textbf{(B1), (B2):} $|C|,|P|,\mathit{old}_n$ unchanged.
    \item \textbf{(B3):} $\mathit{pc}_n'(a)=\textsc{done}\neq\textsc{init}$ and $\mathit{pc}_n(a)=\textsc{mid}\neq\textsc{init}$, so $Q$'s clause implies $P$'s as they are both vacuously true; other ids are unchanged.
    \item \textbf{(B4):} $\mathit{pc}_n(id)=\textsc{mid}$ makes $P$'s clause vacuous for $a$, while $Q$ constrains it, so $Q$'s clause implies $P$'s clause; other ids are unchanged.
\end{itemize}
Hence $Q\implies P$ and $\incre{P}{skip}{Q}$ is a valid Incorrectness triple.
 
\paragraph{Case 3:} $\textsf{begin\_read}(a)$.

\emph{Implementation steps:}

Assumes $\mathit{pc}_n(a)=\textsc{init}$; adds $a$ to $A$; sets $\mathit{MS}(a)\leftarrow C$, $\mathit{YS}(a) \leftarrow P$, $\mathit{pc}_n(a)\leftarrow\textsc{mid}$.

\emph{Corresponding specification steps:}

$c=\mathtt{read}(a)\sqcup\mathtt{skip}$.

\begin{align*}
    P &= (|C|\le x\le|C|+|P|) \\
    &\quad\;\land\; (\mathit{old}_s\subseteq\mathit{old}_n) \\
    &\quad\;\land\; (\forall id.\ \mathit{pc}_s(id)\neq\textsc{init}\implies \mathit{pc}_n(id)\neq\textsc{init}\wedge i_n(id)=i_s(id))\\
    &\quad\;\land\; (\forall id.\ \mathit{pc}_n(id)=\textsc{done}\implies \mathit{pc}_s(id)=\textsc{done}\wedge o_n(id)=o_s(id)),\\
    Q &= (|C|\le x'\le|C|+|P|) \\
    &\quad\;\land\; (\mathit{old}_s'\subseteq\mathit{old}_n) \\
    &\quad\;\land\; (\forall id.\ \mathit{pc}_s'(id)\neq\textsc{init}\Rightarrow \mathit{pc}_n'(id)\neq\textsc{init}\wedge i_n(id)=i_s(id))\\
    &\quad\;\land\; (\forall id.\ \mathit{pc}_n'(id)=\textsc{done}\Rightarrow \mathit{pc}_s'(id)=\textsc{done}\wedge o_n(id)=o_s'(id)),
\end{align*}
where $\mathit{pc}_n'(a)=\textsc{mid}$ and $|C|,|P|$ are unchanged. Split $Q=Q_{\mathtt{s}}\vee Q_{\mathtt{r}}$ on $\mathit{pc}_s'(a)$, with $Q_{\mathtt{s}}=Q\wedge\mathit{pc}_s'(a)=\textsc{init}$ and $Q_{\mathtt{r}} =Q\wedge\mathit{pc}_s'(a)=\textsc{done}$.
 
Since $\incre{P}{skip}{P}$ is a valid Incorrectness triple due to the \emph{unit} rule, by the \emph{consequence} rule it suffices to show that $Q_{\mathtt{s}} \implies P$ for the skip choice: with $\mathit{pc}_s'(a) = \mathit{pc}_s(a) = \textsc{init}$, $P$'s (B3) is vacuous for $a$, and $\mathit{pc}_n(a)=\textsc{init}$ makes $P$'s (B4) vacuous for $a$; $|C|,|P|,x,old_s$ are unchanged, so the remaining clauses are identical. Hence, $Q_{\mathtt{s}} \implies P$ and $\incre{P}{skip}{Q_{\mathtt{s}}}$ is a valid Incorrectness triple.
 
$\mathtt{read}(a)$ assumes $\mathit{pc}_s(a) = \textsc{init}$ and sets $o_s(a)\leftarrow x$, $\mathit{pc}_s'(a) \leftarrow \textsc{done}$, so
\[
 Q'=\{\tau'\mid \mathit{pc}_s'(a)= \textsc{done},\ o_s'(a)=x',\ \tau'[\mathit{pc}_s(a){\leftarrow}\textsc{init}] \in P\,\}.
\]
where $\incre{P}{read(a)}{Q'}$ is a valid Incorrectness triple due to the \emph{assume} and \emph{assigment} rules. For $\tau'\in Q_{\mathtt{r}}$, $a$ began at this step, so by Lemma~\ref{lem:bstruct}(iii) it linearized here and $o_s'(a)=x'$; and the pre-image $\tau'[\mathit{pc}_s(a){\leftarrow}\textsc{init},\,o_s(a)\, \text{unset}]$ satisfies $P$ ($x, old_s$ unchanged; (B3) vacuous for $a$; (B4) vacuous for $a$ at $\sigma$). Hence $Q_{\mathtt{r}} \implies Q'$, and by consequence $\incre{P}{read(a)}{Q_{\mathtt{r}}}$ is a valid Incorrectness triple.
 
As a result of the \emph{choice} rule, $\incre{P}{\mathtt{read}(id)\sqcup
\mathtt{skip}}{Q_{\mathtt{r}}\vee Q_{\mathtt{s}}}=\incre{P}{c}{Q}$ is a valid Incorrectness triple.
 
\paragraph{Case 4:} $\textsf{begin\_increment}(a)$.

\emph{Implementation step:}

Assumes $\mathit{pc}_n(a)=\textsc{init}$, $a\notin P\cup C$; adds $a$ to $P$; inserts $|C|+|P|=c+p+1$ into $\mathit{old}_n$; adds $a$ to $\mathit{YS}(t)$ for all $t\in A$; sets $\mathit{pc}_n(a)\leftarrow\textsc{mid}$.

\emph{Corresponding specification steps:}

$c=(\textsf{increment}(a)\sqcup\mathtt{skip}) ;\ \mathtt{flush}$, where $\mathtt{flush}=(\mathtt{read}(t) \sqcup\mathtt{skip});_{t\in A}$.

\begin{align*}
    P &=(c\le x\le c+p) \\
    &\quad\;\land\; (\mathit{old}_s\subseteq\mathit{old}_n) \\
    &\quad\;\land\; (\forall id.\ \mathit{pc}_s(id)\neq\textsc{init}\implies \mathit{pc}_n(id)\neq\textsc{init}\wedge i_n(id)=i_s(id))\\
    &\quad\;\land\; (\forall id.\ \mathit{pc}_n(id)=\textsc{done}\implies \mathit{pc}_s(id)=\textsc{done}\wedge o_n(id)=o_s(id)),\\
    Q &= (c\le x'\le c+p+1) \\
    &\quad\;\land\; (\mathit{old}_s'\subseteq\mathit{old}_n') \\
    &\quad\;\land\; (\forall id.\ \mathit{pc}_s'(id)\neq\textsc{init}\implies \mathit{pc}_n'(id)\neq\textsc{init}\wedge i_n(id)=i_s'(id)) \\
    &\quad\;\land\; (\forall id.\ \mathit{pc}_n'(id)=\textsc{done}\implies
      \mathit{pc}_s'(id)=\textsc{done}\wedge o_n(id)=o_s'(id)),
\end{align*}
where $\mathit{pc}_n'(a)=\textsc{mid}$, $|C'|=c$, $|P'|=p+1$, and $\mathit{old}_n'=\mathit{old}_n\cup\{c+p+1\}$. Split $Q=Q_{\mathtt{s}}\vee Q_{\mathtt{i}}$ on $\mathit{pc}_s'(a)$.
 
\emph{Skip branch $\;(Q_{\mathtt{s}}=Q\wedge\mathit{pc}_s'(a)=\textsc{init})$.} Here $\mathtt{skip};\mathtt{flush}$ reduces to $\mathtt{skip}$: with $x$ unchanged, no read reaches a new target, so every $\mathtt{flush}$ choice is $\mathtt{skip}$. We show $Q_{\mathtt{s}}\subseteq P$. By Lemma~\ref{lem:bstruct}(ii) at $\sigma'$, $a$ unlinearized gives $x'=c+ |\{b\in P:\mathit{pc}_s'(b)=\textsc{done}\}|\le c+p$, so (B1) of $P$ holds; (B2) follows from $\mathit{old}_s'=\{0,\dots,x'\}\subseteq\{0,\dots,c+p\}= \mathit{old}_n$; B3, B4 for $a$ are vacuous at $\sigma$; the rest are inherited. By the \emph{consequence} rule, $\incre{P}{skip}{Q_{\mathtt{s}}}$ is a valid Incorrectness triple.
 
\emph{Increment branch $\;(Q_{\mathtt{i}}=Q\wedge\mathit{pc}_s'(a)=
\textsc{done})$.} We derive $[P]\,\textsf{increment}(a);\mathtt{flush}\,
[Q_{\mathtt{i}}]$ by the seq rule through the intermediate

\begin{align*}
    M &=(c{+}1\le x\le c+p+1) \\
    &\quad\;\land\; (\mathit{old}_s\subseteq\mathit{old}_n') \\
    &\quad\;\land\; (\mathit{pc}_s(a)=\textsc{done})\\
    &\quad\;\land\; (\forall id.\ \mathit{pc}_s(id)\neq\textsc{init}\implies \mathit{pc}_n'(id)\neq\textsc{init}\wedge i_n(id)=i_s(id))\\
    &\quad\;\land\; (\forall id.\ \mathit{pc}_n'(id)=\textsc{done}\implies \mathit{pc}_s(id)=\textsc{done}\wedge o_n(id)=o_s(id)),
\end{align*}
i.e. $R_1(\sigma')$ after $a$ has linearized but before the reads, read with $\sigma'$ parameters (so (B3) for $a$ holds, as $\mathit{pc}_n'(a)=\textsc{mid}\neq\textsc{init}$).
 
$\incre{P}{\textsf{increment}(a)}{M}$: $\textsf{increment}(a)$ sends $x$ to $x{+}1$ and $\mathit{old}_s$ to $\mathit{old}_s\cup\{x{+}1\}$ and $\mathit{pc}_s(a)$ to $\textsc{done}$. For $\tau_M\in M$, its pre-image has $x{=}x_M{-}1\in[c,c+p]$ (so (B1) of $P$ holds) and $\mathit{pc}_s(a)= \textsc{init}$ ((B3), (B4) vacuous for $a$ at $\sigma$); thus the triple is a valid Incorrectness triple.
 
$[M]\,\mathtt{flush}\,[Q_{\mathtt{i}}]$: $\mathtt{flush}$ is the sequence $(\mathtt{read}(t)\sqcup\mathtt{skip});_{t\in A}$. Applying the seq and choice rules once per $t\in A$ exactly as in the read branch of Case~3 -- taking the $\mathtt{read}(t)$ alternative for the active reads with $o_s'(t)=x'$ (these linearize at the new count $x'$, by Lemma~\ref{lem:bstruct}(iii), since $M$ has $x=x'$) and $\mathtt{skip}$ for the rest -- yields $\incre{M}{\mathtt{flush}}{Q_{\mathtt{i}}}$. The \emph{sequencing} rule then gives $\incre{P}{\textsf{increment}(a);\mathtt{flush}}{Q_{\mathtt{i}}}$ as a valid Incorrectness triple.
 
Distributing the leading choice, $c=(\textsf{increment}(a);\mathtt{flush})\sqcup(\mathtt{skip};\mathtt{flush})$, the \emph{choice} rule combines the two branches into $\incre{P}{c}{Q_{\mathtt{i}} \vee Q_{\mathtt{s}}}=\incre{P}{c}{Q}[P]$, which is a valid Incorrectness triple.
 
\medskip\noindent
Since the corresponding Incorrectness triple is valid for every implementation operation, $R$ is a backward simulation by Theorem~\ref{thm:backwardAsIncorrectnesses}.